\documentclass[12pt,a4paper]{article}

\usepackage{amsmath,amsfonts,latexsym,amssymb}
\usepackage{verbatim,amsthm,curves,graphics}
\usepackage{mathrsfs}
\usepackage[dvips]{graphicx}
 
\usepackage{latexsym}
\usepackage{mathrsfs} 
\usepackage{epsfig}
\usepackage{longtable}
\usepackage{slashed}
\usepackage[all]{xy}
\usepackage{graphicx,tikz}
\usetikzlibrary{arrows}
\usepackage{setspace}

\usepackage{graphicx,color}
\usepackage{array}
\usepackage{tikz}
\usetikzlibrary{shapes.geometric}
\usetikzlibrary{decorations.pathmorphing}
\usetikzlibrary{arrows,shapes,positioning}
\usetikzlibrary{decorations.markings}
\tikzstyle arrowstyle=[scale=1]
\tikzstyle directed=[postaction={decorate,decoration={markings,mark=at position .65 with {\arrow[arrowstyle]{stealth}}}}]
\tikzstyle reverse directed=[postaction={decorate,decoration={markings,mark=at position .65 with {\arrowreversed[arrowstyle]{stealth};}}}]
\def\ben{\begin{equation}}
\def\een{\end{equation}}
\def\bena{\begin{eqnarray}}
\def\eena{\end{eqnarray}}

\def\d{{\rm d}}

\def\S{{\mathcal S}}

\def\C{\mathcal{C}}

\def\A{{\mathcal A}}

\def\cG{{\cal G}}
\def\F{{\cal F}}

\def\cP{{\mathcal P}}

\def\T{{\mathcal T}}

\def\INT{{\mathcal V}}

\newcommand{\SOd}{{\rm SO}(d)}

\def\SOdm{{\rm SO}(d-1)}
\def\SOdmm{{\rm SO}(d-2)}
\def\Ed{{\rm E}(d)}

\def\path{\operatorname{Pa}}

\def\root{\mathcal{R}}

\newcommand{\RR}{{\mathbb R}}
\newcommand{\CC}{{\mathbb C}}
\newcommand{\ZZ}{{\mathbb Z}}

\def\D{{\mathcal D}}

\newcommand{\bh}{\bar h}

\renewcommand{\O}{{\mathcal O}}

\renewcommand{\max}{{\operatorname{max}}}
\renewcommand{\min}{{\operatorname{min}}}

\theoremstyle{theorem}
\newtheorem{thm}{Theorem}

\newtheorem{prop}{Proposition}

\newtheorem{defn}{Definition}

\title{Action principle for OPE}
\author{
Stefan Hollands\thanks{\tt stefan.hollands@uni-leipzig.de}\:
\\ \\
{\it Institut f\"ur Theoretische Physik, Universit\"at Leipzig
} \\
{\it Br\"uderstr. 16, Leipzig, D-04103, Germany
}  \\
}

\begin{document}
\maketitle

\begin{abstract}
We formulate an ``action principle'' for the operator product expansion (OPE) describing how a given OPE coefficient changes under 
a deformation induced by a marginal or relevant operator. Our action principle involves no ad-hoc regulator or renormalization and applies to general (Euclidean)
quantum field theories. It implies a natural definition of the renormalization group flow for the OPE coefficients and of coupling constants. When applied to the case of conformal theories, the action principle 
gives a system of coupled dynamical equations for the conformal data. The last result has also recently been derived (without considering tensor structures) independently by
Behan (arXiv:1709.03967) using a different argument. Our results were previously announced and outlined at the meetings ``In memoriam Rudolf Haag'' in September 2016 and the
``Wolfhart Zimmermann memorial symposium'' in May 2017.

\end{abstract}


\section{Introduction}

One possible viewpoint of quantum field theory (QFT) is that the operator product expansion \cite{wilson1969,zimmermann} (OPE) defines a theory, just as the 
equations of motion define a classical field theory.  

Informally, the OPE states that
\ben\label{OPE}
\O_{A_1}(x_1) \cdots \O_{A_n}(x_n) = \sum_B  \C_{A_1 \dots A_n}^B(x_1, \dots, x_n) \O_B(x_n), 
\een
which is understood in the sense of an insertion into a correlation function, and where $\{ \O_A \}$ is the set of all composite operators of the theory; for details see 
sec. \ref{sec:OPE}. Here we present a ``variational-'' or ``action principle'' for the OPE which states how the coefficients in this expansion change under a change in a coupling parameter $g$ of the theory. If the OPE coefficients define the theory, then such a 
formula, relating the derivative w.r.t. $g$ of a coefficient with other coefficients at the same value of $g$, 
should clearly exist, and it should depend only on the OPE coefficients, the kind of deformation (given by a relevant or marginal operator), but not on extraneous 
structures such as arbitrary regulators etc. However, it is not clear a priori exactly what form it should take.

	The path integral suggests a formal action principle, but such naive formulas require renormalization, while we are looking for an intrinsic formula not requiring such extraneous procedures. Such a formula was derived only relatively recently in~\cite{Holland:2015tia} building on earlier work \cite{Holland:2014ifa,Holland:2012fs,kellerkopper1993}. It gives an expression for the derivative of an 2-point OPE coefficient w.r.t 
	$g$ in terms of 2- and 3-point OPEs involving the relevant or marginal operator $\INT$ conjugate to $g$. To close the system, these are supplemented by analogous formulas for the derivative of an $n$-point OPE coefficient w.r.t 
	$g$ in terms of 2- and $(n+1)$-point OPEs. By contrast to similar hierarchies of functionals such as the Dyson-Schwinger equations, 
	our equations have the desired property of being manifestly finite, i.e. are not in need of additional renormalisation or regularization. For antecedents of our equations see \cite{guida} which is somewhat similar in spirit but different in detail and works for relevant perturbations and OPE-coefficients of operators of sufficiently low dimensions. 
	
	The purpose of the present paper is to briefly explain the relationship of the action principle to the renormalisation group, and secondly, to apply it in the case of conformal QFTs (CFTs). In those theories, the $n$-point OPE coefficients and correlation functions 
	can, in principle, be written in terms of the conformal data, i.e. the dimensions $\Delta_i$ of the conformal primaries and the structure constants
	$\lambda_{ijk}^\alpha$, where $\O_i, \O_j, \dots$ are conformal primaries and $\alpha$ is a label for independent tensor structures needed if the operators are not scalar. 
	Our action principle is shown to imply the following dynamical system	
	\ben\label{ev3}
	\begin{split}
A_i \frac{\d}{\d g}  \Delta_i &=  \sum_\alpha \D^\alpha_i \lambda_{\INT ii}^\alpha\\
\frac{\d}{\d g}  \lambda^\mu_{jkl} &= \sum_m \sum_{\alpha\beta} A_m \left(
{}^a \T_{jklm}^{\alpha\beta\mu}  \ \lambda_{\INT jm}^\alpha \lambda_{klm}^\beta +
{}^b \T_{jklm}^{\alpha\beta\mu} \ \lambda_{jkm}^\alpha \lambda_{\INT lm}^\beta +
 {}^c \T_{jklm}^{\alpha\beta\mu} \ \lambda_{\INT km}^\alpha \lambda_{jlm}^\beta
 \right), 
\end{split}
\een 
where  $A_i=\lambda_{ii1}$ is a normalization factor for the 2-point function\footnote{This could be set to $1$ at the expense of a change in $\T$.}, and 
the constants $\T, \D$ are of an entirely representation theoretic nature, i.e. in principle determined by 
the dimensions $\Delta_i$ and the tensor structures of the primary fields $\O_i$. We give formulas for them 
in terms of the spinning conformal blocks in general dimensions. Thus, to apply the dynamical system, one has to know these 
quantities as well as the possible tensor structures. Their determination and classification is a rather complicated problem in practice
that is drawing considerable attention in the conformal bootstrap \cite{ElShowk:2012ht} community, 
see e.g. \cite{Costa:2011dw,Costa:2016xah,schomerus,Karateev:2017jgd,Kravchuk:2017dzd}. 
In $d=2$, there are explicit expressions \cite{osborn}, which lead to explicit formulas in terms of hypergeometric functions also given below. 

If all the coefficients $\T, \D$ could be found (with their dependence on the dimensions $\Delta_i$), then the dynamical system
could be used in combination with a Newton iteration to find the flows $\Delta_{i}(g), \lambda_{ijk}(g)$, starting from a reference CFT. 
A particularly interesting case  is when the starting point $g=0$ corresponds to a Gaussian free field, as is the case in the 
$\mathcal N=2,4$ super CFTs in $d=4$ based on gauge theories. 
Furthermore, this strategy may open up a new avenue in the 
long-standing problem of mathematically establishing the existence of non-trivial QFTs in $d=4$ dimensions. To accomplish this, one would need perhaps not 
completely explicit knowledge of the coefficients $\T, \D$, but at least sufficient control for large dimensions and spins. Furthermore, one would have to 
check the axioms for the resulting OPE coefficients (most notably associativity) which are not immediately evident from \eqref{ev3}. In the context of perturbation theory, 
we have given in \cite{Holland:2015tia} an argument how the action principle can be used to derive associativity order-by-order in perturbation theory, but the argument is rather 
complicated. Nevertheless, it offers hope that this can be done.

The system \eqref{ev3} without tensor structures has also recently been derived independently by
	\cite{connor} using a somewhat different argument (with an explicit computation of the constants $\T$ for $d=1$). 
	Our results were announced and outlined before at the meetings ``In memoriam Rudolf Haag'' in Hamburg in 
	September 2016 \cite{Hamburg} as well as at the
``Wolfhart Zimmermann memorial symposium'' in Munich in May 2017 \cite{Munich}. This work is dedicated to the memory of Wolfhart Zimmermann whose work on the OPE \cite{zimmermann}
has been a major inspiration for us to further study this structure in QFT.

\section{General structure of QFT}

\subsection{Euclidean QFT}

Although this paper is not about axiomatic quantum field theory, 
to set the stage, we first recall the basic properties of 
correlation functions in Euclidean QFT, called the ``Osterwalder-Schrader (OS)-axioms'', see e.g. \cite{glimm}. 
The basic idea is to formulate general properties of Euclidean Green's functions, 
perhaps constructed by a properly defined path integral, 
\ben
\langle \O_{A_1}(x_1) \dots \O_{A_N}(x_N) \rangle
=
\int_\phi \O_{A_1}(x_1) \dots \O_{A_N}(x_N) \exp (-S(\phi)) , 
\een
where $S$ is an action such as e.g. the $\phi^4_d$-model
\ben
S = \int ( |\partial \phi|^2 + m^2 \phi^2 + g\phi^4) \d^d x . 
\een
and where the composite fields are expressions of the 
form $\O_A = \partial^{a_1} \phi \cdots \partial^{a_r} \phi$, where $\partial^{a_i}$ is a multi-derivative.
Of course, the precise definition and renormalization of such a path integral is a very 
complicated matter which has been successfully accomplished only so far for $d=2,3$ in a full, 
non-perturbative manner. In $d=4$, only perturbative constructions exist, and in that case the OS-axioms are 
satisfied in the sense of formal series in $g$--to the extent that they make sense for formal series. 

It is not the case that all reasonable quantum field theories will arise from a classical path integral--many 
counterexamples exist e.g. in $d=2$, and presumably also in higher dimensions. Thus, in general, 
$\O_A$ should not be thought of as somehow -- i.e. modulo renormalization -- corresponding to a polynomial 
of a ``basic field'' and its derivatives. Rather, $A$ is an abstract label that incorporates the type of field and 
its tensor- or spinor structure. The set of labels will carry some extra structure e.g. related to hermitian conjugation, 
Bose/Fermi alternative, etc., as will become clear when we need it. In the above example  $A = \{ a_1, \dots, a_r \}$ is 
basically a multi-index.

The expected properties of the correlation functions $\langle \O_{A_1}(x_1) \dots \O_{A_N}(x_N) \rangle$ are as follows:

\begin{enumerate}
\item[e1)] Each $\S_{A_1 \dots A_n}(x_1, \dots, x_n)\equiv \langle \O_{A_1}(x_1) \dots \O_{A_N}(x_N) \rangle$ is real analytic for non-coinciding points, i.e. on the configuration space $M_n=\{ (x_1, \dots, x_n)  \mid x_i \in {\mathbb R}^d, \ \ x_i \neq x_j \ \ \forall i \neq j \}$.

\item[e2)] (Identity) There is a neutral element $A=1$ in our index set (corresponding intuitively to the identity operator $\O_1 = 1$)
characterized by the following condition.
 If for $1 \le k < n$  we have $A_k = 1$, then 
 \ben
 \S_{A_1 \dots A_n}(x_1, \dots, x_n) = \S_{A_1 \dots \hat A_k \dots A_n}(x_1, \dots, \hat x_k, \dots x_n) \ , 
 \een
 where a hat means that the index/point is omitted.
 
 \item[e3)] (Bose/Fermi alternative): There exists an assignment $A \to F_A \in \{ 0, 1 \}$ such that, for any $1 \le i<j \le n$, we have 
 \ben
  \S_{ \dots A_i \dots A_j \dots }(\dots, x_i, \dots, x_j, \dots) =  (-1)^{F_{A_i} F_{A_j}} \S_{ \dots A_j \dots A_i \dots }(\dots, x_j, \dots,  x_i, \dots) \ , 
 \een
 and such that $F_B = \sum F_{A_i} \ {\rm mod} \ 2$. 
 We think of a field $\O_A$ as a Bose field if $F_A=0$ and as a Fermi field if $F_A = 1$. 

\item[e4)] (Star operation) This axiom states that there is an involutive $*$-operation $ A \mapsto A^*$ on the set of labels $A$ such that 
\ben
\overline{\S_{A_1 \dots A_n}(x_1, \dots, x_n)} = \S_{A_n^* \dots A_1^*}(x_n, \dots, x_1), 
\een
where overbar means component-wise complex conjugation. 

\item[e5)] (Euclidean group action). Each composite field $\O_A$ carries a finite-dimensional representation $D$ of $\SOd$ (or its double cover if spinor fields are present).
For any element $(r,a) \in \Ed = \SOd \ltimes \RR^d$ of the Euclidean group, or its cover ${\rm Spin}(d) \ltimes \RR^d$ in the case of fermionic fields, 
\ben
\S_{A_1 \dots A_n}(r x_1+a, \dots, r x_n+a) = \bigotimes_{j=1}^N D_j(r) \S_{A_1 \dots A_n}(x_1, \dots, x_n).
\een

\item[e6)] (OS reflection positivity) This axiom is the replacement of the positive definite Hilbert space condition in Lorentizian QFT. Its precise 
form can be found in \cite{glimm}.
\end{enumerate}

Although the above properties  allow for a reconstruction of a corresponding Lorentzian QFT
by the famous OS-reconstruction theorem, they somehow do not tell us directly how the composite fields are related algebraically, and, similarly, they do 
not capture ``short distance factorization'' properties that one typically expects in QFT. From a pragmatic viewpoint, the (related) essential shortcoming of the OS-axioms is that we cannot meaningfully say how the QFT changes under a variation of a coupling parameter such as $g$ in $\phi^4_4$-theory. From the formal path integral, one is tempted to write:
\ben
\partial_g \langle \O_{A_1}(x_1) \dots \O_{A_N}(x_N) \rangle = \int \d^4 y \
\langle \O_{A_1}(x_1) \dots \O_{A_N}(x_N) \INT(y) \rangle
\een
where $\INT = -\phi^4$ in this case. The problem is that the integral makes no sense due to the non-integrable short distance 
divergences at $y=x_i$, for which there is nothing intrinsic to the OS-axioms that would tell us how to regularize it. 
To capture these aspects, it seems inevitable to introduce  
the operator product expansion (OPE).  

\subsection{Axioms for OPE}\label{sec:OPE}

In this paper, we take the viewpoint that a QFT is {\em defined} by the OPE. Thus, the burden of a construction of a QFT is (at least) to
give {\em all} OPE coefficients $\{\C_{A_1 \dots A_n}^B(x_1, \dots, x_n;g)\}$ for arbitrary
composite field labels $\{A_i\}$ and arbitrary $n$. Informally, the coefficients fulfill \eqref{OPE} in the sense of an insertion into a correlation function.
The precise meaning of this relation is explained in 10) below.
We have included a dependence on a parameter $g$ in the OPE coefficients, which 
is thought of as representing a coupling constant. One may be interested in the theory at only one value of this parameter, or 
one may be interested in the whole family of theories labelled by $g$. Obviously, there is no need to restrict to a one-parameter family, 
and in principle, the parameters could belong to a manifold or even more general mathematical structures such as orbifolds.
  
To get a bigger conceptual picture, we first give a list of axioms that 
one would like the OPE coefficients in a reasonable Euclidean quantum field theory in $d>1$ dimensions to fulfill, 
see~\cite{hollands1, HollandsWald} for details. 
 \begin{enumerate}
 \item (Analyticity) For each choice of $A_1, \dots, A_n, B$, $(x_1, \dots, x_n) \mapsto \C_{A_1 \dots A_n}^{B}(x_1, \dots, x_n)$ is a complex valued real 
 analytic function on the ``configuration manifold'' $M_n=\{ (x_1, \dots, x_n)  \mid x_i \in {\mathbb R}^d, \ \ x_i \neq x_j \ \ \forall i \neq j \}$. 
 
 \item (Identity) There is a neutral element $A=1$ in our index set (corresponding intuitively to the identity operator $\O_1 = 1$)
characterized by the following condition.
 If for $1 \le k < n$  we have $A_k = 1$, then 
 \ben
 \C_{A_1 \dots A_n}^B(x_1, \dots, x_n) = \C_{A_1 \dots \hat A_k \dots A_n}^B(x_1, \dots, \hat x_k, \dots x_n) \ , 
 \een
 where a hat means that the index/point is omitted.

 \item (Bose/Fermi alternative) There exists an assignment $A \to F_A \in \{ 0, 1 \}$ such that, for any $1 \le i<j < n$, we have 
 \ben
  \C_{ \dots A_i \dots A_j \dots }^B(\dots, x_i, \dots, x_j, \dots) =  (-1)^{F_{A_i} F_{A_j}} \C_{ \dots A_j \dots A_i \dots }^B(\dots, x_j, \dots,  x_i, \dots) \ , 
 \een
 and such that $F_B = \sum F_{A_i} \ {\rm mod} \ 2$ for every non-zero OPE coefficient  $\C_{A_1 \dots A_n}^B(x_1, \dots, x_n)$. 
 We think of a field $\O_A$ as a Bose field if $F_A=0$ and as a Fermi field if $F_A = 1$. 

\item (Base point) A fairly obvious 
axiom specifying how the coefficients change under a change of the base point (in our conventions, this is always the last point, $x_n$
for a product of $n$ operators). 

\item (Dimension and scaling) The exists a assignment of ``dimensions'' $A \mapsto \Delta_A \in {\mathbb R}_+$ such that the identity operator has dimension $\Delta_1=0$
 and such that, for any $\delta>0$, we have\footnote{Note that this axioms is consistent with a scaling behavior of the form
 $\C_{AB}^C(x_1,x_2) \sim p(\log |x_{12}|) |x_{12}|^{-\Delta_A-\Delta_B+\Delta_C}$ for any polynomial $p$, as indeed found in perturbation theory where 
 the $\Delta_A=[A]$ coincide with the engineering dimensions and where the degree of $p$ increases with the loop order.
 }
 \ben
\lim_{\lambda \to 0+} \lambda^{\Delta_{1}+\dots + \Delta_n - \Delta_B +\delta} \C_{A_1 \dots A_n}^B(\lambda x_1, \dots, \lambda x_n) = 0 \ . 
 \een
 The axiom connects the abstract notion of dimension with the short-distance behavior of the coefficients. We assume that the theory is ``rational'' in the sense that 
 there is a finite number of fields--i.e. labels $A$--having dimension less than some given number.\footnote{A reasonable strengthened version of this would be that the ``partition function'' $\sum_A q^{\Delta_A}$ converges for sufficiently small $q>0$.}

\item (Star operation) This axiom states that there is an involutive $*$-operation $ A \mapsto A^*$ on the set of labels $A$, and that the OPE coefficients of the 
starred operators are equal to the hermitian conjugates of the un-starred operators under a change of base point $x_1 \leftrightarrow x_n$. 

\item (Descendants) This axiom states first that there is an operation $A \mapsto \partial A$ on the set of labels which we think 
of in the example of the $\phi^4_4$-model as corresponding to the partial derivative of the composite field and applying the ``Leibniz rule''. 
The non-trivial part of this axiom is that $\O_{\partial A}$ should behave in the OPE just as $\partial \O_A$, i.e. 
\ben
 \C_{A_1 \dots \partial A_k \dots A_n}^B(x_1, \dots, x_n) = \frac{\partial}{\partial x_k}  \C_{A_1 \dots A_n}^B(x_1, \dots, x_n) \ . 
\een
It is also natural to demand that the dimension of an operator to be increased by one under this operation, $\Delta_{\partial A}=\Delta_A+1$. 

\item (Associativity)
To formulate the associativity of the OPE in its most basic form, 
some ``rationality'' assumption is needed stating that the index set should be countable. 
The strongest possible form of the associativity condition is perhaps the following. Let $1< M < N$. 
\ben
 \C_{A_{1}\ldots A_{N}}^{B}(x_{1},\ldots, x_{N})=\sum_{C}  \C_{A_{1}\ldots A_{M}}^{C}(x_{1},\ldots, x_{M})  \C_{C A_{M+1}\ldots A_{N}}^{B}(x_{M},\ldots, x_{N})
\een
holds on the domain defined by 
$\xi \equiv \frac{ \max_{1\leq i\leq M}|x_{i}-x_{M}|}{ \min_{M<j\leq N}|x_{j}-x_{M}| } < 1 $, the the sum over $C$ is required to be absolutely convergent.
We can, and should, also impose a more general condition for convergence of ``nested'' expansions on corresponding domains, see \cite{hollands1}
for details. 

For $N=3$ the 
ratio $\xi = \frac{|x_1-x_2|}{|x_2-x_3|}<1$ geometrically signifies to what extent the triangle of points 
$x_1, x_2, x_3$ is degenerate. Weaker versions of the axiom would only require that the
OPE coefficient factorizes approximately if a subset of points (here $x_1, \dots, x_M$) get closer to each other than to the remaining points, i.e. 
in the limit as $\xi \to 0$.

\item (Euclidean invariance) An obvious axiom expressing that the OPE coefficients are covariant under the action of the Euclidean group in $d$ dimensions, analogous to that for the correlation functions.

\item (Convergence): 
Let $1< M < N$. 
\ben
\langle \O_{A_1}(x_1) \dots \O_{A_N}(x_N) \rangle 
=\sum_{C}  \C_{A_{1}\ldots A_{M}}^{C}(x_{1},\ldots, x_{M}) \langle \O_{C}(x_M) \dots \O_{A_N}(x_N) \rangle
\een
holds on the domain defined by $\xi \equiv \frac{ \max_{1\leq i\leq M} |x_{i}-x_{M}|}{ \min_{M<j\leq N}|x_{j}-x_{M}| } < 1 $, where the sum over $C$ is required to be absolutely convergent for fixed $\{x_i\}$ inside the domain.
This axiom is very similar looking to the associativity condition and becomes identical if we have
$
\C_{A_{1}\ldots A_{N}}^{1}(x_{1},\ldots, x_{N}) = \langle \O_{A_1}(x_1) \dots \O_{A_N}(x_N) \rangle . 
$
One can always apply a field-redefinition (see below) to satisfy this condition, 
essentially by imposing $\langle \O_A(x) \rangle = 0$, i.e. a vanishing 1-point function, for the redefined fields. This condition is not always natural\footnote{Such a normalization is not always natural because it might have non-analytic/non-smooth behavior e.g. near phase transition points
in coupling constant space, or in the presence of boundaries (Casimir effect).}, but commonly imposed in CFTs (see below).
 \end{enumerate}
 
It is natural to consider two quantum field theories to be {\em equivalent} if their OPE coefficients differ by a ``field-redefinition''. 
A field redefinition is abstractly a linear map $Z$ from the space of fields to itself satisfying certain properties. 
Since we are working throughout with a fixed basis of fields $\O_A$ labelled by the indices $A$, we may equivalently say that a field redefinition corresponds
to a ``matrix'', $Z_A{}^B$, the basic properties of which are dictated by the above axioms: The identity axiom 2) implies $Z_1{}^B = 1$ if $B=1$ and $Z_1{}^B=0$ otherwise. 
The Bose-Fermi alternative 3) implies that $Z_A{}^B = 0$ unless $F_A+F_B = 0 \ \ {\rm mod} \ 2$. The dimension and scaling axiom 5) implies that $Z_A{}^B=0$ unless 
$\Delta_A \ge \Delta_B$.  The star operation axiom 6) implies that $\overline{Z_A{}^B} = Z_{A^*}{}^{B^*}$. The Euclidean invariance axiom 9) implies that, viewed as a tensor over $\mathbb{R}^d$, 
$Z_{A}{}^B$ should be invariant under the group action of the Euclidean group in $d$ dimensions. 

\begin{defn}\label{defn1}
Given a matrix $Z_A{}^B$ satisfying these properties, we say that two systems 
$\C{}_{A_1 \dots A_n}^{A_\root}$ and $\widehat \C{}_{A_1 \dots A_n}^{A_\root}$ are equivalent under $Z_A{}^B$ if there holds
\ben
Z_{B_\root}{}^{A_\root} \widehat \C_{A_1 \dots A_n}^{B_\root}(x_1, \dots, x_n) = \C{}_{B_1 \dots B_n}^{A_\root}(x_1, \dots, x_n) \ Z_{A_1}{}^{B_1} \cdots Z_{A_n}{}^{B_n}
\een
for all $n$, all indices, and all $(x_i) \in M_n$. Note that the implicit sums over $B_1, B_2, \dots$ etc. are {\em finite} because $Z_A{}^B=0$ unless 
$\Delta_A \ge \Delta_B$, so there are no convergence issues. 
\end{defn}

The transformation formula given in the definition corresponds to a redefinition
\ben
\widehat \O_A = \sum_{B: \Delta_A \ge \Delta_B} Z_A^B \ \O_B \ 
\een
of the composite fields.

In the context of renormalized perturbation theory, the freedom to redefine fields 
is in one-to-one correspondence with the freedom of choosing different ``renormalization conditions''~\cite{zimmermann1}. There, it is also natural to 
view the OPE coefficients as functions of the coupling constant(s), (in our case $g$) and to allow also a diffeomorphism acting on $g$, as well 
as a dependence of $Z_A{}^B$ on $g$. 
The renormalization group then also has a natural formulation in terms of field-redefinitions, see section~\ref{sec:RG}.

\subsection{Do the axioms hold?}

Of course, it remains to be seen in models to what extent we can actually fulfill this ``wish list.'' 
There are by now many models with strictly relevant interaction for which the OS-axioms and further results about clustering, Borel summability of the 
perturbation series, etc. have been shown. Apart from many CFT models in $d=2$, these include in particular the  $P(\phi)_2$-models in 2 dimensions--scalar fields with stable polynomial interaction--the $\phi^4_3$-model as well as QED$_3$. For a summary of these developments up to the time of printing see the classic book of Glimm-Jaffe \cite{glimm} and references therein. In a noteworthy development, Dimock \cite{Dimock1,Dimock2,Dimock3} has recently picked up an approach of Balaban and reconsidered some of these models and methods. Comprehensive references to Balaban's work can also be found there. Balaban's method is in principle also applicable to treat the non-perturbative UV renormalization of models with marginal interaction such as YM$_4$, but it is not completely clear to the author what has and what has not 
been achieved in his certainly very impressive series of papers in this direction. The full construction of YM$_4$ without infra-red cutoffs is a famous open problem. For the GN$_2$-model the situation is better and an effective action has been constructed non-perturbatively by \cite{gawedzki}, and also by \cite{Feldman:1986ax}. These results basically cover the full OS-axioms.

For models with marginal interactions, there are also many very good results in perturbation theory establishing the OS-axioms and many much more detailed properties in the sense of formal series. The most interesting theories are perhaps the YM$_4$ models, which have by now been treated in 
full mathematical detail in \cite{Frob:2015uqy, Efremov:2017sqi} based on the method of RG-flow equations \cite{polchinski1984,wetterich1993} and the BRST-BV technique. 
Related prior works for the massless  $\phi^4_4$-model are \cite{guidakopper2011, guidakopper2015}, where several important new techniques for massless theories were introduced in this context. 

Concerning the OPE, there are considerably fewer results. In the case of CFTs, L\" uscher and Mack \cite{luescher1976, mack1976} have demonstrated the 
OPE of two operators inside a 3-point function, but not the general version of the OPE and associativity given above. For a convincing theoretical physics style argument that this is the case, see \cite{pappadopuloetal2012}. 
In the perturbative setting, convergence of the OPE (with ``smeared'' spectator fields) for non-conformally invariant models was first shown by 
\cite{Hollands:2011gf} in the massive $\phi^4_4$-model to arbitrary but fixed order in perturbation theory. Later, this was generalized to 
the massless case \cite{Holland:2014pna}, with good error bounds for the remainder in a finite OPE, and also to perturbative YM$_4$-theory \cite{Frob:2016mzv}. 
The OPE and associativity condition in the form given above was established in the massive and massless $\phi^4_4$ model 
by \cite{Holland:2015tia} to arbitrary but finite orders in perturbation theory. This version of the OPE has also been established 
for $d=2$ CFTs within the Vertex Operator Algebra framework \cite{huangkong2005}.

\section{Action principle: How the OPE changes under deformations}

\subsection{General QFTs}

We will now present the ``action principle'' to construct the OPE coefficients, derived in~\cite{Holland:2015tia} in the context of perturbative $\phi^4_4$-theory. The derivation~\cite{Holland:2015tia} of the action principle generalizes straightforwardly to any theory with power counting renormalizable interaction, such as the ${\rm GN}_2$-model, and with modifications due to local gauge symmetry, also to YM$_4$-theory ~\cite{Frob:2016mzv}. It describes how an OPE coefficient, which is itself a complicated function not only of the points $x_i$ but also of the coupling constant $g$, changes when we vary $g$. 

To write it down, we first use a graphical notation. We draw an OPE coefficient 
$\C_{A_1 \dots A_n}^B(x_1, \dots, x_n)$ 
as

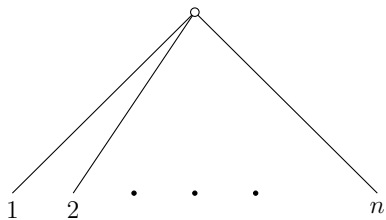
\begin{figure}[h!]
\begin{center}
\begin{tikzpicture}[scale=.8, transform shape] 
	    \filldraw (-1, 0) circle (1pt);
	    \filldraw (0, 0) circle (1pt);
	    \filldraw (1, 0) circle (1pt);
	    
	    \draw (0, 3) -- (- 3 , 0);
	    \draw (0, 3) -- (- 2 , 0);
	    \draw (0, 3) -- (3 , 0);

	    \filldraw[color=black, fill=white] (0, 3) circle (2pt) node[above] {$$};
	    \node[below] at (-3, 0) {$ 1 $};
	    \node[below] at (-2, 0) {$ 2 $};
	    \node[below] at (3, 0) {$ n $};      
	\end{tikzpicture}
\end{center}
\end{figure}

Next, we draw a concatenation of OPE coefficients 
$
\C_{A_1C}^B(x_1,x_n) \C_{A_2 \dots A_n}^C(x_2, \dots, x_n)
$
as

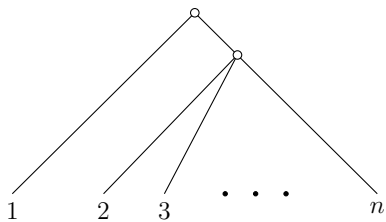
\begin{figure}[h!]
\begin{center}
\begin{tikzpicture}[scale=.8, transform shape]
\filldraw (7.5, 0) circle (1pt);
	    \filldraw (8, 0) circle (1pt);
	    \filldraw (8.5, 0) circle (1pt);
	    
	    \draw (7, 3) -- (4, 0);
	    \draw (7, 3) -- (10, 0);
	    \draw (7.7, 2.3) coordinate (a1) -- (5.5, 0) coordinate (c1);
	    \draw (7.7, 2.3) coordinate (a1) -- (6.5, 0) coordinate (c2);
	    
	    \coordinate (c1) at (intersection of a1--c1 and a1--c2); 
	    \filldraw[color=black, fill=white] (c1) circle (2pt);

	    \filldraw[color=black, fill=white] (7,3) circle (2pt) node[above] {$$};  
	    \node[below] at (4, 0) {$ 1 $};
	    \node[below] at (10, 0) {$ n $};  
	    \node[below] at (5.5, 0) {$2$}; 
	    \node[below] at (6.5, 0) {$ 3 $};
	    	\end{tikzpicture}		
\end{center}
\end{figure}

We also write

\pagebreak

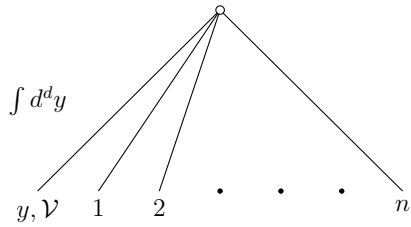
\begin{figure}[h!]
\begin{center}
\begin{tikzpicture}[scale=.8, transform shape]  
\node at (-3,1.5) {$\int d^dy$};
	    \filldraw (0, 0) circle (1pt);
	    \filldraw (1, 0) circle (1pt);
	    \filldraw (2, 0) circle (1pt);
	    	    
	    \draw (0, 3) -- (-3, 0);
	    \draw (0, 3) -- (-2, 0);
	    \draw (0, 3) -- (-1, 0); 
	    \draw (0, 3) -- (3, 0);  
	   
	    \filldraw[color=black, fill=white] (0, 3) circle (2pt) node[above] {$$};
	    \node[below] at (-3, 0) {$ y,\INT $};
	    \node[below] at (-2, 0) {$ 1 $}; 
	    \node[below] at (-1, 0) {$ 2 $};
	    \node[below] at (3, 0) {$ n $}; 
\end{tikzpicture}		
\end{center}
\end{figure}

to mean that

\begin{itemize}
\item 
$\INT$ denotes the ``deformation'', given e.g. by $-\phi^4$ in the $\phi^4_4$-model\footnote{Since this model is expected to exist only in the sense of formal 
series in $g$, our action principle is only expected to hold in that sense in this model.} or by $-(\bar \psi \psi)^2$ in the 
GN$_2$-model,  
\item
$\int \d^dy=$ integral over $\{|y-x_n|<L\}$. 
\item 
$L$ is length scale that is part of the definition of the theory; the formula is valid for points such that $|x_i-x_j|<L$. 
\end{itemize}

The ``action principle'' for OPE coefficients is:
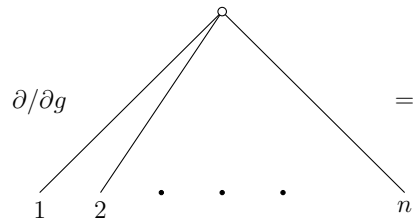
\begin{figure}[htbp]
\begin{center}
\begin{tikzpicture}[scale=.8, transform shape] 
\node at (-3,1.5) {$\partial/\partial g$};
\node at (3,1.5) {$=$};
	    \filldraw (-1, 0) circle (1pt);
	    \filldraw (0, 0) circle (1pt);
	    \filldraw (1, 0) circle (1pt);
	    
	    \draw (0, 3) -- (- 3 , 0);
	    \draw (0, 3) -- (- 2 , 0);
	    \draw (0, 3) -- (3 , 0);

	    \filldraw[color=black, fill=white] (0, 3) circle (2pt) node[above] {$$};
	    \node[below] at (-3, 0) {$ 1 $};
	    \node[below] at (-2, 0) {$ 2 $};
	    \node[below] at (3, 0) {$ n $};      
	\end{tikzpicture}
\end{center}
\caption{Functional equation, left side. The tree represents a coefficient $\C_{A_1 \dots A_n}^B(x_1, \dots, x_n)$}
\label{fig:feql}
\end{figure}
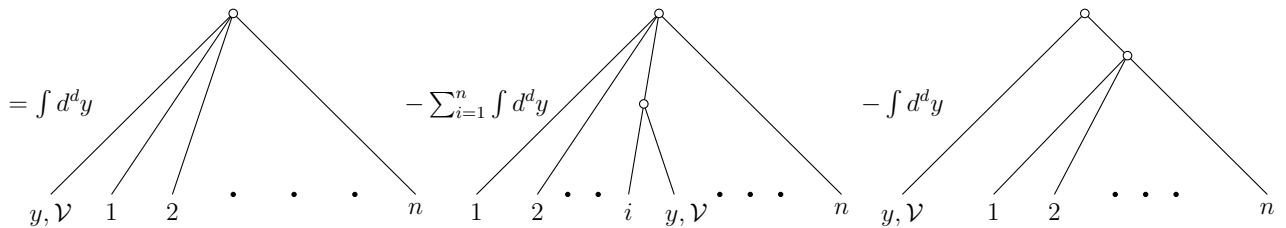
\begin{figure}[htbp]
\begin{center}

\begin{tikzpicture}[scale=.8, transform shape]  
\node at (-10,1.5) {$=\int d^d y$};
	    \filldraw (-7, 0) circle (1pt);
	    \filldraw (-6, 0) circle (1pt);
	    \filldraw (-5, 0) circle (1pt);
	    	    
	    \draw (-7, 3) -- (-10, 0);
	    \draw (-7, 3) -- (-9, 0);
	    \draw (-7, 3) -- (-8, 0); 
	    \draw (-7, 3) -- (-4, 0);  
	   
	    \filldraw[color=black, fill=white] (-7, 3) circle (2pt) node[above] {$$};
	    \node[below] at (-10, 0) {$ y,\INT $};
	    \node[below] at (-9, 0) {$ 1 $}; 
	    \node[below] at (-8, 0) {$ 2 $};
	    \node[below] at (-4, 0) {$ n $}; 
\node at (-3,1.5) {$-\sum_{i=1}^n \int d^d y$};
	    \filldraw (1, 0) circle (1pt);
	    \filldraw (1.5, 0) circle (1pt);
	    \filldraw (2, 0) circle (1pt);
	     \filldraw (-1, 0) circle (1pt);
	     \filldraw (-1.5, 0) circle (1pt);
	    
	    \draw (0, 3) -- (-3, 0);
	    \draw (0, 3) -- (-2, 0);
	    \draw (0, 3 ) coordinate (a) -- (-.5, 0) coordinate (c);
	    \draw (0, 3) -- (3, 0); 
	    \draw ( - .25 , 1.5 ) coordinate (e) -- (0.25, 0) coordinate (e1);
	    
	    \coordinate (c1) at (intersection of a--c and e--e1); 
	    \filldraw[color=black,fill=white] (c1) circle (2pt);

	    \filldraw[color=black, fill=white] (0, 3) circle (2pt) node[above] {$$};  
	    \node[below] at (-3, 0) {$ 1 $};
	    \node[below] at (-2,0) {$ 2$};
	    \node[below] at (-.5, 0) {$ i $};
	    \node[below] at (3, 0) {$ n $};  
	    \node[right] 	 at (- 1, 1.5) {}; 
	    \node[below] at (.45, 0) {$ y,\INT $};      
\node at (4,1.5) {$-\int d^d y$};
\filldraw (7.5, 0) circle (1pt);
	    \filldraw (8, 0) circle (1pt);
	    \filldraw (8.5, 0) circle (1pt);
	    
	    \draw (7, 3) -- (4, 0);
	    \draw (7, 3) -- (10, 0);
	    \draw (7.7, 2.3) coordinate (a1) -- (5.5, 0) coordinate (c1);
	    \draw (7.7, 2.3) coordinate (a1) -- (6.5, 0) coordinate (c2);
	    
	    \coordinate (c1) at (intersection of a1--c1 and a1--c2); 
	    \filldraw[color=black, fill=white] (c1) circle (2pt);

	    \filldraw[color=black, fill=white] (7,3) circle (2pt) node[above] {$$};  
	    \node[below] at (4, 0) {$ y, \INT $};
	    \node[below] at (10, 0) {$ n $};  
	    \node[below] at (5.5, 0) {$1$}; 
	    \node[below] at (6.5, 0) {$ 2 $};
	    	\end{tikzpicture}
		
\end{center}
\caption{Functional equation, right side. The composite trees represent concatenations of coefficients, e.g. the rightmost tree means 
$\sum_C \C_{A_1 \dots A_n}^C(x_1, \dots, x_n) \C_{\INT C}^B(y,x_n)$.}
\label{fig:feqr}
\end{figure}


In symbols, the action principle reads as follows for $n=2$ points:
\ben\label{recurorigpert2}
\begin{split}
& \partial_g \C_{A_1 A_2}^B(x_1,x_2) =    \int\limits_{\epsilon<|y-x_i| \leq L}\d^{d}y\,  \C_{\INT A_1 A_2}^B(y,x_1,x_2)\\
&\qquad -  \sum_{\Delta_C \leq \Delta_1} \int\limits_{\epsilon<|y-x_{1}| < L}\d^{d}y\,  \C_{\INT A_1}^C(y,x_1)   \C_{ CA_2}^B(x_1,x_2)
-  \sum_{\Delta_C \leq \Delta_2} \int\limits_{\epsilon<|y-x_{2}| < L}\d^{d}y\,  \C_{\INT A_2}^C(y,x_2)   \C_{ A_1C}^B(x_1,x_2)
\\
& \qquad - \sum_{\Delta_C<\Delta_B}  \int\limits_{\epsilon<|y-x_{2}|< L}\d^{d}y\,  \C_{A_1 A_2}^C(x_1,x_2)   \C_{\INT C}^B(y,x_2) \, ,
\end{split}
\een
where the limit $\epsilon \to 0^+$ is understood in the end. 
 The terminology ``action principle'' is due to the fact that a derivative with respect to $g$ corresponds to an ``insertion'' of 
the interaction operator $\INT$ in the first term on the right side.  

The action principle~\eqref{recurorigpert2} -- together with the hierarchy of similar relations with more points $\{x_i\}$ -- has the following features:
\begin{enumerate}
\item It {\em only} involves the OPE-coefficients, but no correlation functions.
\item The integral over the ``insertion point'', $y$, is shown to be {\em absolutely convergent}, and in particular free of any
of the seemingly unavoidable UV-divergences in QFT, in the sense that the limit $\epsilon \to 0^+$ can be taken without further 
regulators. Here the idea is that, due to associativity, the terms in the second line of \eqref{recurorigpert2} 
should {\it precisely cancel the UV divergences} of the first line
that would appear in the limit $\epsilon \to 0^+$. 
Similarly, the third line cancels any UV divergences that would appear when $L \to \infty$, again due to associativity.
\item The formula makes no reference to perturbation theory. Thus, given any theory with given OPE coefficients at some value of the 
coupling, we can chose a relevant or marginal interaction $\INT$ and define the deformed theory abstractly as a solution to our action principle. 
\item For $n$ points and more couplings, the generalization is given below in \eqref{recurorigpert2}.
\end{enumerate}

A perturbative version of the action principle was proved in~\cite{Holland:2015tia,Holland:2014ifa} for the $\phi^4_4$-model, 
starting from a definition of the renormalized OPE coefficients using flow 
equations. The rigorous statement is that if we make a power series expansion in $g$ of the OPE coefficients, schematically $\C = \sum g^r \C_r$, around the free Gaussian fixed point $g=0$, then the action principle holds to all orders in $g$. In other words, it was proved that:

\begin{thm}
Let $\INT=-\phi^4$ and $d=4$. Then to any order $r$ in perturbation theory there holds the recursion relation~\eqref{recurorigpert2} (and its obvious variants for 
$n$ point OPE coefficients) where the limit $\epsilon \to 0$ is understood, and where $\Delta_A=[A]$ in the formula is given by the engineering dimension of $\O_A$ at the Gaussian fixed point ($g=0$). When the terms on the right side are written under a single integral sign, then the integral is absolutely convergent
for $\epsilon \to 0$, i.e. all non-integrable singularities of the total integrand cancel out between the different terms. 
\end{thm}

We view this as convincing evidence for the general validity of the action principle~\eqref{recurorigpert2} even beyond perturbation theory (i.e. formal series), in models where a non-perturbative theory actually exists. The $\phi^4_4$-model is not believed to exist non-perturbatively except for the Gaussian 
fixed point ($g=0$), but we expect the action principle to be true in models where it does, such as in the GN$_2$-model, or the exactly marginal flows described below. At any rate, {\em we will from now on assume the validity of the action principle at the non-perturbative level.}

In case we have more  marginal or relevant couplings denoted by $g^a$ corresponding to marginal or relevant operators $\O_a$ in the action, then a generalized 
action principle can be derived by the same method as in the case of one marginal coupling. The formula is now 
\ben\label{recurorigpert4}
\begin{split}
& \frac{\partial}{\partial g^a} \C_{A_1 \dots A_n}^B(x_1, \dots, x_n) =    \int\limits_{\epsilon<|y-x_i| \leq L}\d^{d}y\,  \C_{a A_1 \dots A_n}^B(y,x_1,\dots, x_n)\\
&\qquad -  \sum_{i=1}^n \sum_{\Delta_C \leq \Delta_i} \int\limits_{\epsilon<|y-x_{i}| < L}\d^{d}y\,  \C_{a A_1}^C(y,x_i)   \C_{ A_1 \dots C \dots A_i}^B(x_1, \dots, x_n)
\\
& \qquad - \sum_{\Delta_C<\Delta_B}  \int\limits_{\epsilon<|y-x_{n}|< L}\d^{d}y\,  \C_{A_1\dots A_n}^C(x_1, \dots, x_n)   \C_{a C}^B(y,x_n) \, ,
\end{split}
\een
where the limit $\epsilon \to 0^+$ is understood in the end, and where we have generalized formula \eqref{recurorigpert2} to $n$ points.

\subsection{Action principle and renormalization group for OPE} \label{sec:RG}

\subsubsection{Geometry of field redefinitions}

As we have emphasized, we are free in general to make 
$g$-dependent redefinitions of the fields $\widehat \O_A^{} = \sum Z^B_A(g) \O_B^{}$, where
for massless theories we should demand on dimensional grounds that $Z^B_A$ has non-vanishing entries for $\Delta_A=\Delta_B$. 
For the so redefined fields, we get new OPE coefficients $\widehat \C_{A_1 \dots A_n}^{B}$ as defined in Definition \ref{defn1}. 

The re-defined coefficients will satisfy a modified action principle. The relation with the action principle before the redefinition is best explained in 
geometric terms. Let us denote by $g^a = (g^1, \dots, g^n)$ the parameters of the theory which we loosely think of as associated with 
terms of the form $\sum g^a \O_a$ in the action--if the theory has one. Each operator $\O_a$ is marginal or relevant, i.e. $\Delta_a \le d$.
Let 
\ben
\A \equiv \A^A_{Ba}(g) \d g^a
\een 
be a ``connection'' in field space, non-zero only for $\Delta_A \le \Delta_B$. A curve $c^a(\tau)$ in coupling constant space is called geodesic 
if its tangent $\dot c^a(\tau)$ is parallel transported to itself under the connection $\A$; in equations
\ben
\ddot c^a = -\A^a_{bc}(c) \, \dot c^b \dot c^c. 
\een
If we consider curves starting at $g=0$, then they are uniquely given (locally) once we give $\widehat g^a=\dot c^a(0)$. Geodesic normal coordinates around $g^a=0$ are defined by assigning to $\widehat g^a$ the value $c^a(1)$ in coupling constant space (this is the ``exponential map'' in geometric terms, 
$g={\rm Exp}_0^\A \widehat g$). Define $Z(\widehat g)$ 
to be the parallel transport (holonomy) along a geodesic from $0$ to $\widehat g$ under $\A$, in formulas 
\ben
Z(\hat g) = T \, {\rm exp} \int_0^1 \A_a(c(\tau)) \dot c^a(\tau) \, \d \tau.
\een  
The components of $Z$ are denoted by $Z^A_B$. They are non-zero only if $\Delta_B \le \Delta_A$. We use this $Z^A_B$ in order to define 
a new set of OPE coefficients $\widehat \C$ as in Definition \ref{defn1}. These new coefficients are viewed as functions of $\widehat g$. 
Then it is easily seen that the action principle for $\widehat \C$ is the same as that (see \eqref{recurorigpert4}) for $\C$, except that we have to replace the hatted coefficients 
everywhere, and we have to make the replacement on the left side:
\ben\label{redef}
\frac{\partial}{\partial g^a} 
\C_{A_1 \dots A_n}^{B}
\to 
\frac{\partial}{\partial \widehat g^a} 
\widehat \C_{A_1 \dots A_n}^{B} +
\sum_i \A^{C_i}_{A_i a}
\widehat \C_{A_1 \dots C_i \dots A_n}^{B}
- \A^{B}_{C} \widehat \C_{A_1 \dots A_n}^{C}
\een
where the hatted quantities and $\A$ on the right side are now viewed as functions of $\hat g$. So, geometrically, the 
partial derivative $\frac{\partial}{\partial g^a}$ (parallel transport with trivial connection) is replaced by parallel transport 
along the geodesic tangent to $\frac{\partial}{\partial \hat g^a}$ with the connection $\A$. In practice, we will use field redefinition to remove such 
unwanted terms involving $\A$ below.

In case we have a perturbatively defined theory, then all statements must be understood in the sense of formal series in $g^a$. 

\subsubsection{Renormalization group (RG) flow}

The action principle is closely related to a version of the RG-flow. The discussion is again completely general, but to be concrete, we can think of a model like the ${\rm GN}_2$-model or the $\phi^4_4$ model. The OPE coefficients are defined as formal series in that case $g$ and the action principle is understood order-by-order in powers of $g$. To avoid a cumbersome discussion of relevant interactions, we assume that there is only one marginal coupling constant, and that 
the renormalized mass is taken to be $m=0$ (at each order in perturbation theory). 
Now we ask how the OPE-coefficients change if we change $L$ to some (larger) value $\widehat L= e^t L$, where $t >0$. Let the new OPE coefficients 
(determined again by the action principle \eqref{recurorigpert2} and its obvious generalization to $n$ points) be denoted by $\widehat \C_{A_1 \dots A_n}^{B}$. Let us define the RG ``time'' as usual by
\ben
t=\log \frac{\widehat L}{L} \ ,
\een
assuming that $L,\widehat L$ are so large that $L,\widehat L>|x_{ij}|$ for any pair of points considered in an OPE coefficient.

The answer to this question is closely related to the geometric definition of the field redefinitions just explained. More precisely, the following proposition expresses the renormalization group:

\begin{prop}
In a massless perturbative theory with only one marginal interaction $\INT$, 
there exist formal power series $f(t,g) = g + \sum_{r\ge2} f_r(t) g^r$ and $Z_A^B(g,t)=\delta_A^B + \sum_{r \ge 1} Z_r{}_A^B(t) g^r$  such that the 
OPE-coefficients $\widehat \C_{A_1 \dots A_n}^{A_\root}(g)$ for IR-cutoff $\widehat L$ are equivalent to the OPE-coefficients $\C_{A_1 \dots A_n}^{A_\root}(f(t,g))$ for IR-cutoff $L$
under the field redefinition 
\ben
\widehat \O_A = \sum_B Z^B_A(g,t) \O^{}_B \ , 
\een
where  $Z_A^B(g,t)$ vanishes if $[B]\neq [A]$. (Here $[A]$ is the engineering dimension of a field $\O_A$). 
\end{prop}
\begin{proof}
The proof of this proposition for massless $\phi^4$-theory follows directly from the arguments given 
in sec.~V of~\cite{Holland:2015tia}. Apply proposition~1 of that reference
twice, for $L$ and $\widehat L$. For a small mass $m$, it follows from the formulas given in the paper that $Z, f$ are invariant if we 
rescale simultaneously $L,\widehat L,m$. Therefore, in the limit as $m \to 0$, which is argued in the paper to exist, $Z,f$ can only depend on the 
ratio $L/\widehat L$, hence on $t$, as well as $g$. There is nothing special about the 
argument for this particular theory in the sense that the proof generalizes as long as we only have 
one marginal operator. For several marginal operators one obtains a flow in an $n$-dimensional space of couplings. 
\end{proof}

We stress that  each coefficient $f_r(t)$ in the formal series for $f(g,t)$ has 
a well-defined, finite value. Similar remarks apply to the 
formal series $Z_A^B(g,t)$. Then, by applying the previous proposition twice for $t$ and $t'$ it therefore follows that
the ``cocycle'' identities 
\ben
Z(t+t',g) = Z(t',g(t,g))Z(t,g) \ , \quad 
f(t+t',g) = f(t',f(t,g)) \  
\een
hold in the sense of formal series. It is standard that the cocycle conditions imply the existence of a formal series 
$\beta(g)$ and a matrix-valued formal series $\Gamma_A^B(g)$ in $g$ with vanishing entries for $[A]\neq [B]$ satisfying 
\ben\label{differential}
\partial_t f(t,g) = \beta(f(t,g)) \ , \quad \partial_t Z(t,g) = \Gamma(f(t,g)) Z(t,g) \ . 
\een
For the ${\rm GN}_2$ model, the first terms in the formal series $\beta(g) = \sum \beta_r g^r$ are found to be $\beta_0=\beta_1=0$ and
\ben\label{decay}
 \beta_2 = -(N-1)/\pi, \quad \beta_3 = (N-1)/(2\pi^2)  \ . 
\een  
These results, which we do not derive here, are
equivalent to the usual ``1-loop'' and ``2-loop'' calculations of the beta function in the GN-model.

In a theory with vanishing beta coefficients to all orders, the cocycle relation becomes the 1-parameter group condition
$Z(t+t',g) = Z(t',g)Z(t,g)$.  Thus we can write 
$Z(t,g) = \exp t\Gamma(g)$. In a theory satisfying the OS-positivity condition, one can show that $\Gamma$ is semi-simple, i.e. 
diagonalizable with real eigenvalues (no non-trivial Jordan-blocks). The eigenvalues $\gamma_A(g)$ correspond to the anomalous 
dimensions. More properly, we should say that the true  dimensions are the engineering dimensions\footnote{Engineering dimensions 
are defined only at the Gaussian fixed point (free field). In a scalar field theory in $d=4$, the engineering dimension of $\O_A$ is 
the number of $\phi$-factors plus the number of derivatives.} $[A]$ plus the 
anomalous dimensions, $\Delta_A(g) = [A]+\gamma_A(g)$. 
Indeed, in a basis of fields in which $\Gamma(g)$ is diagonal, our action principle and version of the renormalization group gives
\ben\label{cov}
\C_{A_1 \dots A_n}^B(\lambda x_1, \dots, \lambda x_n) 
= \lambda^{-\Delta_{A_1}-\cdots-\Delta_{A_n}+\Delta_B} \C_{A_1 \dots A_n}^B(x_1, \dots, x_n)
\een
in the sense of formal series.  
This can easily be derived  
using that the action principle is trivially invariant under 
a simultaneous rescaling of $L \to \lambda L$ 
and $x_i \to \lambda x_i$ at each perturbation order, up to the power counting factor $\lambda^{-[A_1]-...-[A_n]+[B]}$.
Thus, the action principle is scale-covariant. Since a multiplicative change of $L$ corresponds to a multiplicative 
change of the fields by $Z(\log \lambda,g) = \lambda^{\Gamma(g)}$ with eigenvalues $\lambda^{\gamma_A(g)}$, the result follows.

The RG flow is also nicely intertwined with the geometry underlying the field redefinitions. Suppose we subject the OPE coefficients to a field redefinition given by a matrix $\zeta^A_B(g)$ and 
the marginal or relevant coupling parameters to the corresponding change $g^a \mapsto \hat g^a(g) \equiv \phi(g)$ induced by the exponential map (see above). Then under this change, the vector field $\beta^a$ and the matrix of anomalous dimensions $\Gamma^A_B$ change according to  
\ben
\Gamma(g) \to \widehat \Gamma(g) = \zeta(g)^{-1} \Gamma(\phi(g)) \zeta(g) -  \zeta(g)^{-1} \widehat \beta \cdot \d \zeta(g) \ ,
\quad \beta(g) \to \widehat \beta(g) =  \phi_* \beta(g)  \ ,  
\een 
where $\phi_*$ denotes the pull-back of a vector field. In perturbation theory, these statements are understood in the sense of formal series in $g$. 
It is a simple exercise in manipulating formal series to show that, if $\beta_2$ is non-vanishing, there 
always exist formal series $\phi$ and $\zeta_A^B$ such that $\widehat \beta_r = 0$ for all $r \ge 4$ and $\widehat \Gamma_r{}^{B}_A = 0$ for all $r \ge 2$. 
This corresponds to the statements in the physics literature that ``only $\beta_2, \beta_3$ and $\Gamma_1$ are universal''.

\subsection{Action principle for CFTs}

\subsubsection{Conformal blocks, correlation functions, and tensor structures}

Conformal field theories are (Euclidean, in this paper) QFTs whose correlation functions 
are invariant not only under the Euclidean group but even under the conformal group, see e.g. \cite{rychkov2} for a recent, hands-on, review. 
More precisely, consider the group of all globally defined orientation preserving conformal transformations $g$ of ${\mathbb R}^d$ (here $d \ge 2$), 
\ben
g^* \delta_{\mu\nu} = \Omega(x)^2 \delta_{\mu\nu}, 
\een
where $\Omega(x)$ is called the conformal factor and $\delta_{\mu\nu}$ is the metric on Euclidean space ${\mathbb R}^d$. 
As is well-known, the conformal group is isomorphic to  a covering group\footnote{Here we will only 
be concerned with a sufficiently small neighbourhood of the identity of this group, so the covering is in this sense unnecessary.} of 
the identity component of ${\rm SO}(d+1,1)$. It obviously contains the (cover of the) Euclidean 
group $\Ed$, and in particular the rotation subgroup $\SOd$ or ${\rm Spin}(d)$, for which the conformal factor is $\Omega=1$. 
The most important new transformations are the dilations $d(\lambda): x \mapsto \lambda x$, which 
have conformal factor $\Omega = \lambda$. 

To make these axioms, one assumes that within the set $\{\O_A\}$ of local fields (or rather, mathematically speaking in the abstract set of labels $A$)
there are distinguished so-called ``primary'' fields. These fields are labelled in the following by a lower case Roman index from the middle of the alphabet, i.e. 
the set of primary fields is denoted by $\{\O_j\}$. Each such primary field may have a tensor- or spinor index structure (subsumed in the index $j$) -- alternatively speaking, 
they take values in some finite-dimensional irreducible representation $D_j$ of $\SOd$, or its
cover ${\rm Spin}(d)$ if we want to treat fermions, on a {\em finite dimensional} real or complex vector space\footnote{Such representations are in correspondence with a Young-tableau indicating the symmetry properties of the tensor, or by a set of spin labels $(s_1, \dots, s_{\lfloor d/2 \rfloor})$.} $V_j$. Furthermore, it is assumed that any other local field $\O_A$ can be obtained 
as  primary field can be obtained by a linear combination of derivatives of primary fields. Such fields are called ``descendants''. 

The covariance axiom under the Euclidean group is now replaced by the following axiom:

\begin{enumerate}
\item[e5')] Correlation functions of primary fields satisfy
\ben
\langle \O_{i_1}(gx_1) \dots  \O_{i_N}(gx_N) \rangle = \prod_{j=1}^N \Omega(x_j)^{\Delta_j} \bigotimes_{j=1}^N D_j(R(x_j,g)) 
\langle \O_{i_1}(x_1) \dots  \O_{i_N}(x_N) \rangle, 
\een
where $gx$ is the usual action -- if defined -- of a group element $g \in {\rm SO}(d+1,1)$ on a point $x\in \mathbb{R}^d$. 
The relation is required to hold for $g$ in a neighborhood $U \subset {\rm SO}(d+1,1)$ of the identity such that for any smooth path $g_t$ in $U$
connecting $g \in U$ with the identity, none of the points $g_t x_i$ goes through infinity. 
Furthermore, $R(x,g)$ is a $\SOd$ element obtained by decomposing the Jacobian of the transformation 
$g$ as $\partial (gx)/\partial x = \Omega(x,g) R(x,g)$, 
and $\Delta_j$ is the dimension of the primary field. Since any other field can be obtained from primary fields 
by applying derivatives, we have a corresponding (more complicated) transformation formula also for the descendant fields. 
\end{enumerate}

As is well known, the axiom gives drastic simplifications of the correlation functions at the 2- and 3-point level. Let us assume to simplify notations that all fields are hermitian (i.e. $i^*=i$). It is a consequence of OS-positivity that, at the 2-point level, we
can choose the primary fields such that
\ben
\langle \O_i(x_1) \O_j(x_2) \rangle = A_i \delta_{ij} t_{ij}(x_1, x_2) 
|x_{12}|^{-2\Delta_i}   , 
\een
where $x_{12} \equiv x_1-x_2$,  where $A_i \in \mathbb C$ and where $t_{ij}$ is an ``invariant tensor structure''\footnote{We could absorb $A_i$ into 
the tensor structure or a field redefinition. However, it is more convenient for us not to do this here, see footnote \ref{footnote14}. If we want, we can always set $A_i$ to $1$ by a field redefinition in the very end.}. In general, an $N$-point invariant tensor structure is 
a map from $M_N$ to the tensor product representation space $\otimes_k V_{i_k}$ (with each $V_{i}$ a complex vector space carrying a {\em finite dimensional, irreducible} representation $D_i$ of $\SOd$) such that for all $g \in  {\rm SO}(d+1,1)$ in some open neighborhood of the identity -- depending on the points $x_i$ as in e5') -- we have
\ben\label{action}
t_{i_1 \dots i_N}(gx_1, \dots, gx_N) =  \bigotimes_{j=1}^N D_j(R(x_j,g)) t_{i_1 \dots i_N}(x_1, \dots, x_N). 
\een
For $N=2$ points, we can construct the tensor structures as follows\footnote{For a more detailed exposition of the following ``colinear frame''-type arguments, see e.g. the recent paper \cite{Kravchuk:2016qvl} and references therein. }. First, we note that 2 points $x_1, x_2$ can be 
transformed by a conformal transformation to 
two fixed points, say $0,e$ with $|e|=1$ on a fixed line in $\RR^d$. The remaining conformal transformations leaving these two points fixed form a 
subgroup of ${\rm SO}(d+1,1)$ isomorphic to\footnote{
At the level of Lie-algebras, if we denote the generators $\frak{so}(d,1)$ by their standard notation $M_{\mu\nu}, P_\mu, K_\mu, D$ and if $e=(1,0,\dots,0)$,
then ${\mathfrak e}(d-1)$ is generated by $M_{ij}, 2K_i - M_{1i}$, where $i=2, \dots, d$. 
} ${\rm E}(d-1) = {\rm SO}(d-1) \rtimes \RR^{d-1}$. Thus, the transformation formula tells us that $t_{ij}(0,e)$ lies in the space 
of invariant tensors under the group action \eqref{action}, namely
\ben
\label{2pt1}
t_{ij}(0,e) \in (V_i \otimes V_j)^{{\rm E}(d-1)}   
\een
where the superscript denotes the set of tensors invariant under $D_i(R(0,g)) \otimes D_j(R(e,g))$ with $g$ ranging over a 
sufficiently small neighborhood of the unit element in the group ${\rm E}(d-1)$.
This space can be seen to be 1-dimensional if $V_i \cong V_j$ as a representation. 
So, in the 2-point function, only one tensor structure can appear for fixed $i$, which in the following 
we fix once and for all\footnote{In particular, when $\O_i$ is scalar, we take $t_{ii}=1$.}.

As an example, consider a totally symmetric tensor operator $\O^{\mu_1 \dots \mu_r}$ of rank $r$, 
i.e. $V_i =V_j={\rm S}^r \RR^d$, which is actually a reducible representation of $\SOd$. 
The elements in the space \eqref{2pt1} are linear combinations of
\ben\label{comb}
\begin{split}
t^{\mu_1 \dots \mu_r,\nu_1 \dots \nu_r}(0,e) =& \sum_{P,Q \in S_r}
(\delta^{\mu_{P(1)}\nu_{Q(1)}} - 2e^{\mu_{P(1)}} e^{\nu_{Q(1)}}) \cdots (\delta^{\mu_{P(p)}\nu_{Q(p)}} - 2e^{\mu_{P(p)}} e^{\nu_{Q(p)}}) \\
&
\times \delta^{\mu_{P(p+1)} \mu_{P(p+2)}} \cdots \delta^{\mu_{P(r-1)} \mu_{P(r)}}
\delta^{\nu_{Q(p+1)} \nu_{Q(p+2)}} \cdots \delta^{\nu_{Q(r-1)} \nu_{Q(r)}}, 
\end{split}
\een
for $1 \le p \le r$, and the form for general $x_1 \neq x_2$ is then found using the transformation law \eqref{action}. An irreducible $V_i$ 
is obtained by considering only tensors in ${\rm S}^r \RR^d$ which are trace-free in any pair of indices from either 
$\{\mu_1, \dots, \mu_r\}$ or $\{\nu_1, \dots, \nu_r\}$. The corresponding 2-point tensor structure is then 
precisely the up to scaling unique totally trace-free linear combination of \eqref{comb}. The 2-point structure $t^{\mu_1 \dots \mu_r,\nu_1 \dots \nu_r}(x_1,x_2)$
is obtained via \eqref{action} choosing any group element $g \in {\rm SO}(d+1,1)$ such that $g(0)=x_1, g(e)=x_2$. 

For $r=2$ (which in the case $\Delta_i=d$ leads to the 2-point function of the stress tensor), this gives for instance
\ben
t^{\mu_1 \mu_2, \nu_1 \nu_2}(x_1,x_2) = \left( \delta^{\mu_1(\nu_{1}} - 2\frac{x^{\mu_1}_{12} x^{(\nu_{1}}_{12}}{|x_{12}|^2} \right)
\left( \delta^{\nu_2)\mu_{2}} - 2\frac{x^{\nu_2)}_{12} x^{\mu_{2}}_{12}}{|x_{12}|^2} \right) - \frac{1}{d} \delta^{\mu_1 \mu_2} \delta^{\nu_1 \nu_2}, 
\een
in agreement with standard formulas.

Similarly, for $N=3$ points, we can transform an arbitrary configuration $x_1, x_2, x_3$ to three distinguished points on a line, say
$0,e,\infty$, where $\infty$ denotes the point at infinity on the compactified line $\RR e$. The fixed point group consists of those 
conformal transformations in the subgroup $\SOdm$ leaving the line $e\RR$ fixed. Therefore, we similarly have
\ben
t_{ijk}(0,e,\infty) \in (V_i \otimes V_j \otimes V_k)^{\SOdm}. 
\een
For $N=4$ points, we can transform an arbitrary configuration $x_1, x_2, x_3, x_4$ to four distinguished points on a fixed plane line, say
$0,e,\infty, x$, where $x$ denotes the point in some fixed plane containing $e$. Such a plane is left invariant by a corresponding subgroup 
$\SOdmm$. Therefore, we similarly have
\ben\label{trafomm}
t_{ijkl}(x,0,e,\infty) \in (V_i \otimes V_j \otimes V_k \otimes V_l)^{\SOdmm}. 
\een
Thus, an invariant 4-point tensor structure gives rise to a function of $x$ valued in the invariant tensors. In fact, since we need precisely two 
real numbers to say where we are on our real plane, 
we may say that an invariant 4-point tensor structure gives rise to a function from a distinguished 2-plane in $\RR^d$ to in the invariant tensors. The converse to these statements also holds true, i.e. we may reconstruct invariant tensor structures for 2, 3 or 4 points by the invariants described above. 

The 3-point function of primary fields can be written in terms of such 3-point structures, namely 
 \ben\label{3pt}
 \begin{split}
\langle \O_i(x_1) \O_j(x_2) \O_k(x_3) \rangle = &
|x_{12}|^{-\Delta_i-\Delta_j+\Delta_k} 
|x_{23}|^{-\Delta_j-\Delta_k+\Delta_i}
|x_{13}|^{-\Delta_k-\Delta_i+\Delta_j}\\
&\times \sum_\alpha \lambda^\alpha_{ijk} \ t_{ijk}^\alpha(x_1, x_2, x_2) , 
\end{split}
\een
where $\lambda^\alpha_{ijk} \in \CC$ are called the structure constants, and where $t^\alpha_{ijk}$ is some chosen basis of 3-point tensor structures
from $(V_i \otimes V_j \otimes V_k)^{\SOdm}$. 
Thus, up to 3 points, all we need to know are the dimensions $\Delta_i \in \RR_+$ and the structure constants $\lambda^\alpha_{ijk} \in \CC$. In fact, 
in view of the OPE axiom, all correlation functions are determined by these ``conformal data''. 

Using the form of the 3-point function, one finds that OPE of two primary fields also has a particularly simple form in any CFT. 
It is given by 
\ben
\O_i(x_1) \O_j(x_2) = \sum_k \sum_\alpha \lambda^\alpha_{ijk} \cP^\alpha_{ijk}(x_{12}, \partial_2) \ \O_k(x_2) , 
\een
where the tensor operators $\cP^\alpha_{ijk}(x_{12}, \partial_2)$ are fixed pseudo-differential operators (depending only on the 
tensor structures associated with $i,j,k$ and $\alpha$) that can be found explicitly in principle. When formally expanded out in 
$\partial_2$, they generate an infinite series of contributions from descendant fields. The main difference of the OPE in CFTs compared to the 
case of general QFTs as described above is thus that the OPE coefficients for primaries and descendants are not independent (we also get 
the OPE of two arbitrary descendants). Applying two successive OPEs $(ij)(kl)$ to the 4-point function gives 
\ben\label{4pt}
\langle \O_i(x_1) \O_j(x_2) \O_k(x_3) \O_l(x_4) \rangle = \sum_m \sum_{\alpha,\beta} \lambda_{ijm}^\alpha \lambda_{klm}^\beta A_m^{}
\ \cG_{ijkl,m}^{\alpha\beta}(x_1, x_2, x_3, x_4) , 
\een
(at least) in the domain $\{|x_{12}|, |x_{34}|<|x_{24}|\}$, where the ``conformal blocks'' are defined as
\ben
\cG_{ijkl,m}^{\alpha\beta}(x_1, x_2, x_3, x_4) \equiv \cP^\alpha_{ijm}(x_{12}, \partial_2)  \cP^\beta_{klm}(x_{34}, \partial_4) \langle \O_m^{}(x_2) \O_{m}^{}(x_4) \rangle/A_m. 
\een
These expressions are in principle completely fixed by the representation theory of the conformal group. It is customary to factor out the trivial dependencies 
of the blocks on the points implied by conformal invariance and factor out the tensor 4-point structures. This amounts to writing
\ben
\begin{split}
\cG_{ijkl,m}^{\alpha\beta}(x_1, x_2, x_3, x_4) =& 
\left( \frac{|x_{14}|}{|x_{24}|} \right)^{\Delta_j - \Delta_i}
\left( \frac{|x_{14}|}{|x_{13}|} \right)^{\Delta_k - \Delta_l}
\frac{1}{|x_{12}|^{\Delta_i+\Delta_j} |x_{34}|^{\Delta_k+\Delta_l}}\\
&\times
\sum_\gamma t^\gamma_{ijkl}(x_1, x_2, x_3, x_4) \ \cG_{ijkl,m}^{\alpha\beta\gamma}(u,v) . 
\end{split}
\een
Here $u,v$ are the anharmonic ratios, 
\ben
u= \frac{|x_{12}|^2 |x_{34}|^2}{|x_{13}|^2 |x_{24}|^2}  , \quad
v= \frac{|x_{14}|^2 |x_{23}|^2}{|x_{13}|^2 |x_{24}|^2} , 
\een
and the quantities $\cG_{ijkl,m}^{\alpha\beta\gamma}(u,v)$ are sometimes referred to as the ``spinning blocks''. They are, too, in principle completely determined by 
the representation theory of the conformal group once we give the tensor structures and dimensions associated with $i,j,k,l,m$ and once we fix bases of 3- and 4-point tensor structures $t^\alpha_{ijm}, t^\beta_{klm}, t^\gamma_{ijkl}$. But in practice their determination for general $d$ is a very complicated problem which we do not address here. For the sake of concreteness, let us just quote the well-known result for $d=2$:

\medskip
\noindent
{\bf Example:} In $d=2$, the conformal group is ${\rm SO}(3,1)$. The spin corresponds to a representation of ${\rm SO}(2)$ (for bosons). The representation
spaces $V_j$ are all one-dimensional and the representations are labelled by an integer 
$s \in \ZZ$ (for fermions, $s \in \ZZ/2$). It is convenient to define $h, \bar h$ through the relations $h + \bh = \Delta, h-\bh = s$, where the overbar in $\bh$ does not mean any kind of conjugation. By the general arguments above, there is only one independent invariant 3-point and 4-point tensor structure, so the labels $\alpha, \beta, \dots$ are not needed. Identifying points $x_i \in \RR^2$ with complex numbers and using the notation $h_{ij}=h_i-h_j$ etc., we may write the 4-point tensor structure as
\ben
t_{ijkl}(x_1, x_2, x_3, x_4) =  
\left( \frac{x_{24} \bar x_{14}}{x_{14} \bar x_{24}} \right)^{s_{i}/2-s_j/2}
 \left( \frac{x_{14} \bar x_{13}}{x_{13} \bar x_{14}} \right)^{s_{k}/2-s_l/2}
 \left( \frac{\bar x_{12}}{x_{12}} \right)^{s_{i}/2+s_j/2}
 \left( \frac{\bar x_{34}}{x_{34}} \right)^{s_{k}/2+s_l/2}
\een
and we may write the spinning conformal block as \cite{osborn} 
\ben
\cG_{ijkl,m}(u,v)=
\, {}_2 F_1(h_m-h_{ij}, h_m + h_{kl}, 2h_m; z) \cdot \, {}_2 F_1 (\bh_m-\bh_{ij}, \bh_m + \bh_{kl}, 2\bh_m; \bar z)
z^{h_m} \bar z^{\bh_m}
\een\label{Fdef}
Here, $z \in \CC$ is defined implicitly by $u=\bar z z, v = (1-\bar z)(1-z)$.

\subsubsection{Action principle}

Equipped with this information, we now 
want to investigate what form our action principle takes in the context of CFTs. In general, we must assume that $\INT$ is some operator of dimension
$\Delta_\INT\le d$. However, even if our starting point $g=0$ corresponds to a conformally invariant theory, there is no reason that the flow determined by the action principle will remain conformal. This obviously is not the case when $\Delta_\INT < d$, i.e. for relevant flows. For marginal flows, $\Delta_\INT = d$, this is not necessarily the case either, because the interaction can break conformal symmetry (as for the GN$_2$-model or the $\phi^4_4$-model) in the sense that 
$\Delta_\INT$ will start differing from $d$ as we flow. Thus, we must assume a {\it strictly marginal} perturbation which by definition corresponds to a 1-parameter family of CFTs parameterized by $g$. Such a family is characterized by a 1-parameter
family of conformal data $\{\Delta_i^{}(g), \lambda^\alpha_{ijk}(g)\}$. 
It is clear that our action principle will give an ordinary differential equation for the 1-parameter family of conformal data when applied to $N=3$ points and three primary operators with labels $A=i, B=j, C=k$. 

We now would like to determine what this differential equation is. First, we would like to write down the 3-point OPE coefficients $\C_{ijk}^l$ associated with four primaries. From the OPE $(ijk)l$ in a 4-point function and the fact that the 2-point functions of descendants fall off faster than those of their corresponding primaries, 
one finds
\ben
\lim_{x_4 \to \infty} |x_4|^{2\Delta_l} \langle \O_i(x_1) \O_j(x_2) \O_k(x_3) \O_{l}(x_4) \rangle/A_l = \C_{ijk}^l(x_1, x_2, x_3) t_{ll}^{}(x_3, \infty), 
\een
(no summation over $l$) where $t_{ij}$ is the tensor structure appearing in the 2-point function $\langle \O_i(x_1) \O_j(x_2) \rangle$. Since this tensor structure is invertible, which follows from OS-positivity, we can obviously reconstruct $\C_{ijk}^l$ from the 4-point function which in turn can be reconstructed from the conformal data and spinning conformal blocks. Similarly, we have 
\ben
\lim_{x_3 \to \infty} |x_3|^{2\Delta_l} \langle \O_i(x_1) \O_j(x_2) \O_{l}(x_3) \rangle/A_l = \C_{ij}^l(x_1, x_2) t_{ll}^{}(x_2, \infty), 
\een
so $\C_{ij}^l$ again can be reconstructed from the conformal data. 

We need to write down the 2-point and 3-point structure constants in some more detail. For simplicity, we assume for the moment that 
a normalization (field redefinition) of the operators has been chosen in which $A_i=1$ -- these factors can easily be reinstated and will be in the final formulas. 
For the 2-point OPE coefficient we then get putting $x_1=x, x_2=0$:
\ben\label{2pt}
\C_{ij}^l(x, 0) t_{ll}^{}(0, \infty) = \sum_\alpha \frac{t^\alpha_{ijl}(x, 0, \infty) \lambda^\alpha_{ijl}}{|x|^{\Delta_i+\Delta_j-\Delta_l}} .
\een
For the 3-point OPE coefficient, we get the relevant information from the 4-point function, but the formulas are a bit more complicated, because 
the expansion formulas for the 4-point function in terms of spinning conformal blocks do not converge everywhere.  
Putting $x_1 = x, x_2 = 0, x_3 = e$, we are going find separate expressions for $\C_{ijk}^l(x,0,e)$ in the following domains 
\begin{enumerate}
\item[a)] $\{|x_{12}|< \min(|x_{13}|,|x_{23}|) \}$, equivalent to $\{|x|< \min(|x-e|,1) \}$.
\item[b)] $\{|x_{23}|< \min(|x_{12}|,|x_{13}|) \}$, equivalent to $\{1< \min(|x-e|,|x|) \}$.
\item[c)] $\{|x_{13}|< \min(|x_{23}|,|x_{12}|) \}$, equivalent to $\{|x-e|< \min(|x|,1) \}$.
\end{enumerate}

\medskip
\noindent
a) In this domain \eqref{4pt} is valid, so we can write 
\ben
\begin{split}
\C_{ijk}^l(x,0,e)  t_{ll}^{}(0, \infty) = &\sum_m \sum_{\alpha, \beta, \gamma} \frac{1}{|x|^{\Delta_i+\Delta_j} |x-e|^{\Delta_k - \Delta_l}} \
 \cG_{ijkl,m}^{\alpha\beta\gamma}\left( \frac{|x|^2}{|x-e|^2},  \frac{1}{|x-e|^2}\right) \\
& \times t^\gamma_{ijkl}(x,0,e,\infty)  .
\end{split}
\een

\medskip
\noindent
b) In this domain \eqref{4pt} is valid after permuting $(ijk) \to (jki)$ as well as $(123) \to (231)$, so we can write (assuming 
for simplicity that all primary fields are bosonic)
\ben
\begin{split}
\C_{ijk}^l(x,0,e)  t_{ll}^{}(0, \infty) = &\sum_m \sum_{\alpha, \beta, \gamma} \frac{1}{|x|^{\Delta_i-\Delta_l}} \
 \cG_{jkil,m}^{\alpha\beta\gamma}\left( \frac{1}{|x|^2},  \frac{|x-e|^2}{|x|^2}\right) \\
& \times t^\gamma_{jkil}(0,e, x,\infty) \ \lambda_{jkm}^\alpha \lambda_{ilm}^\beta .
\end{split}
\een

\medskip
\noindent
c) In this domain \eqref{4pt} is valid after permuting $(ijk) \to (ikj)$ as well as $(123) \to (132)$, so we can write (assuming 
for simplicity that all primary fields are bosonic)
\ben
\begin{split}
\C_{ijk}^l(x,0,e)  t_{ll}^{}(0, \infty) = &\sum_m \sum_{\alpha, \beta, \gamma} \frac{1}{|x|^{\Delta_j-\Delta_l} |x-e|^{\Delta_i - \Delta_k}} \
 \cG_{ikjl,m}^{\alpha\beta\gamma}\left( \frac{|x-e|^2}{|x|^2},  \frac{1}{|x|^2}\right) \\
& \times t^\gamma_{ikjl}(x,e,0,\infty) \ \lambda_{ikm}^\alpha \lambda_{jlm}^\beta .
\end{split}
\een 

Now let $\O_i=\INT$ be our exactly marginal perturbation (a primary field of dimension $d$), and let $\epsilon>0$ be small. We define the following complex numbers:

\medskip
\noindent
a) For domain a), we set
\ben
\begin{split}
{}^a \T_{jklm}^{\alpha\beta\mu} =& {\rm P.F.} \sum_{\gamma} \int_{\epsilon < |x|< \min(|x-e|,1)} \d^d x \ \frac{1}{|x|^{d+\Delta_j} |x-e|^{\Delta_k - \Delta_l}} \\
 & \cG_{\INT jkl,m}^{\alpha\beta\gamma}\left( \frac{|x|^2}{|x-e|^2},  \frac{1}{|x-e|^2}\right) \ t^\mu_{jkl}(0,e) \cdot  t^\gamma_{\INT jkl}(x,0,e,\infty)
 \end{split}
\een
where the dot $\cdot$ is the natural hermitian inner product between tensors in $V_j \otimes V_k \otimes V_l$ and $\rm P.F.$ denotes the 
finite part of an asymptotic expansion\footnote{
More precisely, let $f(\epsilon)$ be a function having an asymptotic expansion of 
the form $\sum_{a_i < N} \sum_{b_j > M} A_{i,j} ( \log \epsilon )^{a_i} \epsilon^{b_j}$ for small 
positive $\epsilon$. Then F.P. means that we subtract all of the finitely many divergent terms in this expansion.
Note that the right side indeed has such an expansion (with possibly non-integer powers) due 
to the OPE.} for small $\epsilon$.   

\medskip
\noindent
b) For domain b), we set
\ben
\begin{split}
{}^b \T_{jklm}^{\alpha\beta\mu} =& {\rm P.F.} \sum_{\gamma} \int_{1< \min(|x-e|,|x|) < 1/\epsilon} \d^d x \ \frac{1}{|x|^{d-\Delta_l}}  \\
 &  \cG_{jk\INT l,m}^{\alpha\beta\gamma}\left( \frac{1}{|x|^2},  \frac{|x-e|^2}{|x|^2}\right)   \  t^\mu_{jkl}(0,e) \cdot t^\gamma_{jk\INT l}(0,e, x,\infty). 
 \end{split}
\een

\medskip
\noindent
c) For domain c), we set
\ben
\begin{split}
{}^c \T_{jklm}^{\alpha\beta\mu} =& {\rm P.F.} \sum_{\gamma} \int_{\epsilon < |x-e|< \min(|x|,1)} \d^d x \ \frac{1}{|x|^{\Delta_j-\Delta_l} |x-e|^{d - \Delta_k}} \\
 & \cG_{\INT kjl,m}^{\alpha\beta\gamma}\left( \frac{|x-e|^2}{|x|^2},  \frac{1}{|x|^2}\right) \ t^\mu_{jkl}(0,e) \cdot t^\gamma_{\INT jkl}(x,0,e,\infty)
 \end{split}
\een
Note that the complex numbers ${}^a \T_{jklm}^{\alpha\beta\mu}, {}^b \T_{jklm}^{\alpha\beta\mu}, {}^c \T_{jklm}^{\alpha\beta\mu}$ are in principle 
determined by the representation theory of the conformal group alone, i.e. the spinning blocks and invariant tensor structures, as well as the dimensions of the primary fields. Further note that if we parameterize 
$x=(\Re(z), \Im(z) {\bf z})$ with $z \in \CC, {\bf z} \in S^{d-2}$, then the integrand depends in each of the cases a), b), c)
only on $z,\bar z$ (noting that $|x|^2 = z \bar z, |x-e|^2 = (1-z)(1- \bar z)$), due to the invariance properties of our tensor structures \eqref{trafomm}. 
Then the integration turns into an integration over the corresponding subset of 
$\CC$, and we can effectively write $$\d^d x = (2i)^{-d+2} {\rm vol}(S^{d-2}) (z-\bar z)^{d-2} \d^2 z.$$ Furthermore, the integration domain is
transformed in each case to the ``fundamental\footnote{
It is a fundamental for the action of the permutation group $S_3$ acting by the conformal transformations $1/z, (z-1)/z, z/(z-1)$ which permute the points
$0,1,\infty$.
} domain''
\ben
\F = \{z \in \CC \mid \epsilon < |z| < \min(|z-1|,1) \}
\een
after the change of integration variables $z \to z/(z-1)$ in case a), $z \to 1/z$ in case b) and $z \to (z-1)/z$ in case c). The dot products involving the tensor 
structures can in the variables $z,\bar z$ be written as sums of terms of the form $z^a \bar z^{\bar a}(1-z)^b(1-\bar z)^{\bar b}$  where 
$a - \bar a \in \ZZ, a+\bar a \in \RR$, so at the end of the day we need to perform in each case a), b), c) integrals of the form 
\ben
\T \sim \text{sum of terms of form} \quad {\rm P.F.} \int_\F z^a \bar z^{\bar a}(1-z)^b(1-\bar z)^{\bar b} \cG(z, \bar z) \d^2 z, 
\een
where $\cG$ is some spinning conformal block viewed as a function of the new variables $z, \bar z$, 
and $a,b,\bar a, \bar b$ depend on the dimensions, tensor structures, and the spacetime dimension $d \ge 2$ (in a different way in each of 
the cases a), b), c)).

Now we consider the action principle \eqref{recurorigpert2} for general $A_1, A_2, C$ in CFT, under the assumption that there is only one exactly marginal operator (generating a 
flow of CFTs). Note that, even if $A_1, A_2, C$ correspond to primary fields, the sums in the formula involve primary and descendant fields. 
As we have already said, the role of the terms in the second line of \eqref{recurorigpert2}
is to remove any UV-divergences from the integral in the first line i.e. the contributions that would diverge as $\epsilon \to 0$. 
The terms in the last line are automatically finite as long as $L$ remains finite.
According to our discussion of the renormalization group, we can accommodate a change in  $L$  
by appropriate (diverging, as $L \to \infty $) field redefinitions. 
It is convenient to take $L=1/\epsilon$, because then we can write
\ben\label{recurorigpert5}
\begin{split}
& \partial_g \C_{AB}^C(x_1,x_2) =    {\rm P.F.} \int\limits_{\epsilon<|x_1-y|,|x_2-y| \leq 1/\epsilon}
\d^{d} y\,  \C_{\INT AB}^C(y,x_1,x_2)\\
&\qquad -   \sum_{\Delta_D= \Delta_A} \A_{A}^D  \C_{DB}^C(x_1,x_2)
-  \sum_{\Delta_D = \Delta_B} \A_{B}^D   \C_{AD}^C(x_1,x_2)
 - \sum_{\Delta_D = \Delta_C}  \A_{D}^C \C_{AB}^D(x_1,x_2)   
 \, ,
\end{split}
\een
where P.F. denotes the finite part as $\epsilon \to 0$, whereas $\A_A^B$ are certain (finite) complex constants. These constants arise from the 
RG-flow (since we are varying $L=1/\epsilon$), as well as from the terms in the second and third line of the action principle \eqref{recurorigpert2}.  
It follows in particular that the {\em same} finite constants appear in all three terms. 
The fact that only terms with equal dimension appear in the summations can be seen from dilation covariance since both the finite part P.F. as well as the left side must satisfy the dilation covariance condition \eqref{cov}. 
It follows that the matrices of constants $\A_A^B$ may be absorbed in a further finite field redefinition and a redefinition of the coupling constant, as described in \eqref{redef}. So we learn that, with these implicit field redefinitions, 
\ben\label{recurorigpert6}
\partial_g \C_{AB}^C(x_1,x_2) =    {\rm P.F.} \int\limits_{\epsilon<|x_1-y|,|x_2-y| \leq 1/\epsilon}
\d^{d} y\,  \C_{\INT AB}^C(y,x_1,x_2) \ . 
\een
We now consider the action principle \eqref{recurorigpert6} for $A_1=A_2=i$, with $i$ a label of a primary field, and $C=1$.  
At this stage, we put $x_1 = x, x_2 = 0$, we use that 
$\C_{ij...k}^1 = \langle \O_i \O_j \cdots \O_k \rangle$, we use the explicit form of the 3- and 2-point functions in the CFT in terms 
of the conformal data, and we may use \eqref{2pt}. The resulting explicit integrals can be carried out easily when 
$\O_i$ is scalar normalized so that $t_{ii}=1$ in the 2-point function.
The result is
\ben\label{prev}
\begin{split}
|x|^{-2\Delta_i} \left( 2A_i \log(|x|) \, \frac{\d \Delta_i}{\d g}  + \frac{\d A_i}{\d g} \right) &=  |x|^{-2\Delta_i+d} \ \lambda_{\INT ii} \ 
 {\rm P.F.} \int\limits_{\epsilon < |y|, |x-y| < 1/\epsilon} \frac{\d^d y \, }{
 |x-y|^{d} |y|^{d}}\\
&= \left( \frac{4\pi^{d/2}}{\Gamma(d/2)} \log(|x|) + C \right) |x|^{-2\Delta_i} \, \lambda_{\INT ii}  
\end{split}
\een
where $C$ is a constant\footnote{
\label{footnote14}
The presence of this constant forbids us to set $\d A_i/\d g=0$. Of course, we 
could achieve this by a further field redefinition if we wanted to, but this would change eq. \eqref{recurorigpert6}.} depending only on $d$. 
We conclude that, if $\O_i$ is scalar, then
$$
A_i \frac{\d}{\d g}  \Delta_i = \frac{2\pi^{d/2}}{\Gamma(d/2)} \lambda_{\INT ii} . 
$$
In the general case when $i$ is in a non-trivial representation of the group $\SOd$ the integrals can still be carried out explicitly in principle, but we must take
into account the tensor structures that might appear in the 3-point OPE coefficients $\C_{\INT ii}^1$. The result is now  of the form
\ben\label{ev1}
A_i \frac{\d}{\d g} \Delta_i=  \sum_\alpha \D^\alpha_i \lambda_{\INT ii}^\alpha
\een  
where $\D^\alpha_i$ are the constants that come up in these tensor integrals.

We now choose $A=i, B=j, C=l$, with $i,j,l$ labels of primary fields, and $(x_1, x_2) = (0,e)$. We 
split the remaining integral (integration variable now called $x$) on the right side into the 
regions a), b), c) described before, and use in each region the formulas for the 3-point OPE coefficient
given under a), b), c). Finally, we use the definition of the constants ${}^a \T_{jklm}^{\alpha\beta\mu}, {}^b \T_{jklm}^{\alpha\beta\mu}, {}^c \T_{jklm}^{\alpha\beta\mu}$. On the left side, we use our expression for the 2-point OPE coefficient $\C_{ij}^l(0,e)$. Taking the tensor structures $t^\mu_{jkl}(0,e,\infty)$ to form an orthogonal system in $(V_i \otimes V_j \otimes V_k)^{\SOdm}$ under the natural hermitian inner product on this finite dimensional space denoted by a dot above, and using that the tensor structure 
$t_{ll}(0,\infty)$ appearing in the 2-point function is invertible (by OS-positivity) there results the equation\footnote{
We note in passing that the coefficients $\T$ with arbitrary indices can be related to $6j$-symbols\footnote{$6j$-symbols have recently surfaced in the context of higher dimensional CFTs also in \cite{gadde}.} of the conformal group
$\SOdm$ for the six representations $(V_\INT=\RR, \Delta_\INT=2), (V_i, \Delta_i),$ $(V_j, \Delta_j), (V_k, \Delta_k), (V_l, \Delta_l), (V_m, \Delta_m)$.
Here $(V,\Delta)$ is the (infinite-dimensional) representation of $\SOdm$ described by the $\lfloor d/2 \rfloor$ 
spins encoded in the finite-dimensional representation 
$V$ of $\SOd$ and the dimension $\Delta$.  We shall come back to this in another paper. 
}:
\ben\label{ev2}
\frac{\d}{\d g} \lambda^\mu_{jkl} = \sum_m \sum_{\alpha\beta} \left(
{}^a \T_{jklm}^{\alpha\beta\mu}  \ \lambda_{\INT jm}^\alpha \lambda_{klm}^\beta +
{}^b \T_{jklm}^{\alpha\beta\mu} \ \lambda_{jkm}^\alpha \lambda_{\INT lm}^\beta +
 {}^c \T_{jklm}^{\alpha\beta\mu} \ \lambda_{\INT km}^\alpha \lambda_{jlm}^\beta
 \right) A_m
\een
which is the formula claimed in the introduction.

Consistency with our normalization conditions on the 1-point, 2-point and 3-point function in CFT 
requires that if $i=j$ and $k=1$, then we should have $A_i = \lambda_{ii1}$ which gives the missing evolution equation for $A_i$
as this special case of \eqref{ev2}. Consistency also requires that, if $i \neq j$ and $k=1$, then we should have $\lambda_{ijk}^\alpha=0$ along the flow, and hence 
$\d/\d g \ \lambda_{ijk}^\alpha=0$. This should actually follow from our flow equation and 
is most easily seen for scalar operators $i,j$ (so that there is no need for a tensor structure label $\alpha$) going back to \eqref{recurorigpert6} 
with $A=i, B=j, C=1$. So we only need to check in view of $\C_{ij\INT}^1(x_1,x_2,x_3) = \langle \O_i(x_1) \O_j(x_2) \INT(x_3) \rangle$
and of \eqref{3pt} that  
\ben\label{integral}
{\rm P.F.} \int\limits_{\epsilon < |x_1-y|, |x_2-y| < 1/\epsilon} \frac{\d^d y \, }{
 |x_1-y|^{-\Delta_j+\Delta_i+d} |x_2-y|^{-\Delta_i+\Delta_j+d}}=0
\een 
which indeed holds as long as $x_1-x_2 \neq 0$ and as long as 
\ben\label{non-deg}
\Delta_i - \Delta_j \notin {\mathbb Z}, 
\een
i.e. if a certain non-degeneracy condition is fulfilled by the spectrum of dimensions. This condition is generically only violated if $i$ corresponds to the identity operator $1$ (having dimension 0) and $j$ corresponds to $\INT$ (having by construction dimension $d$). However, in that case \eqref{integral} automatically vanishes anyway. 
The conceptual reason for the vanishing of the integral \eqref{integral} is of group theoretical nature as one may expect. 
The best way to see this is to analytically continue 
the integrand to imaginary $\nu=\Delta_i-\Delta_j$. Then the integral, viewed as a function of $x_1,x_2$, can be seen as a sesquilinear 
bilinear form on the representation space of a spin-0, principal series representation of ${\rm SO}(d+1,1)$ labelled by $\nu$. 
The only such form must be the scalar product itself, 
which corresponds to a delta distribution $\delta^d(x_1,x_2)$ in the principal series, 
and hence vanishes for $x_1 \neq x_2$. By analyticity, this must remain true as long as $\nu$ does not correspond to poles of \eqref{integral}, i.e. as long as the non-degeneracy conditions remains true. 
This line of reasoning also immediately 
shows that the corresponding statement is still true if $i,j$ are not scalar but carry spin. 

Perhaps a better way to state these results is that the field definition required to set ${\mathcal A}=0$ in \eqref{recurorigpert5}
is precisely that which respects our normalization condition of the 1,2,3-point functions in the presence of a non-degeneracy condition. In the absence of such a condition, suitable field redefinition have to be applied, which unfortunately complicates matters considerably. 

If we set $i=j=k=\INT$ in our evolution equations \eqref{ev1}, \eqref{ev2}, 
then we find $\lambda_{\INT\INT\INT}=0$ from the first evolution equation, since we must have $\Delta_\INT=d$ along our flow for
an exactly marginal perturbation. Our second evolution equation then gives a constraint the coefficients $\lambda_{\INT\INT m}$, and by differentiating this constraint 
w.r.t. $g$, further constraints follow which somehow encode that the theory actually {\em has} any exactly marginal deformations. Such constraints were also observed 
previously in \cite{bash}.  

\medskip
\noindent
{\bf Example:} In $d=2$, the conformal group is ${\rm SO}(3,1)$, and we can identify points $x \in \RR^2$ with complex numbers. 
The tensor structures are unique, so the labels $\alpha, \beta, ...$ are superfluous, and the 
equations become somewhat simpler: 
\bena
\frac{\d}{\d g} \lambda_{jkl} &=& \sum_m  \left(
{}^a \T_{jklm}  \ \lambda_{\INT jm} \lambda_{klm} +
{}^b \T_{jklm} \ \lambda_{jkm} \lambda_{\INT lm} +
 {}^c \T_{jklm} \ \lambda_{\INT km} \lambda_{jlm}
 \right) A_m\\
 A_i \frac{\d}{\d g} \Delta_i &=&  2\pi \lambda_{\INT ii}.
\eena
There remains, however, the problem of determining the coefficients $\T$. These were concretely defined in a), b), c) above, so in principle we can find them from the spinning blocks displayed in the previous section. For domain a), we find, for instance, after a change of variables $z=x/(x-1)$:
\ben
\begin{split}
{}^a \T_{jklm} =& 
{\rm P.F.} \int_{\F} \d^2 z \ 
\frac{(z-1)^{h_{kl}-2}}{z^{1-h_j-h_m}}
\frac{(\bar z-1)^{\bh_{kl}-2}}{\bar z^{1-\bh_{j}-\bh_m}} \times \\
& {}_2 F_1(h_m + h_j-1, h_m+h_{kl}, 2h_m; z) \ {}_2 F_1(\bh_m + \bh_j-1, \bh_m+\bh_{kl}, 2\bh_m; \bar z)
\end{split}
\een
where  P.F. denotes the finite part as $\epsilon \to 0$, and where the $\epsilon$-dependent integration domain is
$\F = \{z \in \CC \mid \epsilon < |z| < \min(|z-1|,1) \}$. 
For domain b), we find after a change of variables $z=1/x$:
\ben
\begin{split}
{}^b \T_{jklm} =& 
{\rm P.F.} \int_{\F} \d^2 z \ 
\frac{(z-1)^{-1-h_{k}}}{z^{h_{jk}+h_l+h_m+1}}
\frac{(\bar z-1)^{-1-\bh_{k}}}{\bar z^{\bh_{jk}+\bh_l+\bh_m+1}} \times \\
& {}_2 F_1(h_m + 1 - h_k, h_m + h_{jl}, 2h_m, z) \
 {}_2 F_1(\bh_m + 1 - \bh_k, \bh_m + \bh_{jl}, 2\bh_m, \bar z)
\end{split}
\een
For domain c), we find after a change of variables $z=(x-1)/x$:
\ben
\begin{split}
{}^c \T_{jklm} =& 
{\rm P.F.} \int_{\F} \d^2 z \ 
\frac{(z-1)^{-1+h_j+h_l+h_k}}{z^{1+h_k-h_m}}
\frac{(\bar z-1)^{-1+\bh_j+\bh_l+\bh_k}}{\bar z^{1+\bh_k-\bh_m}} \times \\
& {}_2 F_1(h_m - 1 + h_k, h_m + h_{jl}, 2 h_m, z) \ {}_2 F_1(\bh_m - 1 + \bh_k, \bh_m + \bh_{jl}, 2 \bh_m, \bar z)
\end{split}
\een
The remaining integrals can 
be evaluated for instance using the Mellin-Barnes representation of the hypergeometric function. Using such a representation, 
we are lead to integrals of the form $\int_\F z^a \bar z^{\bar a}(1-z)^b(1-\bar z)^{\bar b} \ \d^2 z$
or rather, their finite part as $\epsilon \to 0$.  
The final result is expressed in terms of a generalized hypergeometric series, which we will explore elsewhere. 
In $d=1$ (``conformal quantum mechanics'' which is formally outside the scope of our framework), 
the integrals  reduce to one-dimensional ones and the formulas become much simpler. In that case, an explicit expression 
has independently been obtained by \cite{connor}. 

\section{Outlook}

The most interesting application of our dynamical equations \eqref{ev1}, \eqref{ev2} for the CFT data of exactly marginal flows would be to 
the $\mathcal N$=2,4 super conformal Yang-Mills theories in $d=4$. Such theories are Gaussian for $g=0$ and 
so trivially soluble, providing thus in principle an initial condition for the dynamical system. Given sufficient information about the coefficients $\T$, which is in principle only kinematical input, might enable one to show by Newton iteration that the dynamical system has a non-perturbative solution for finite $g$, 
thus leading to a mathematical existence result for these theories without any kind or large $N$ limit! 

Unfortunately, as stated our results do not quite apply to this situation, 
because at the Gaussian point the non-degeneracy condition \eqref{non-deg} $\Delta_i - \Delta_j \notin \ZZ$ for the spectrum 
is certainly very far from being fulfilled, as is even the case 
away from the Gaussian point due to primaries with protected dimension. In the absence of the non-degeneracy condition, our evolution equations for the conformal data have to be modified by finite field redefinitions when a ``crossover'' $\Delta_i - \Delta_j \in \ZZ$ occurs as we have explained. This appears to be a non-trivial problem\footnote{I am grateful to S. Rychkov for emphasizing this point to me.}.

An independent technical complication, which to some extent has been addressed by  \cite{Costa:2011dw,Costa:2016xah,schomerus,Karateev:2017jgd,Kravchuk:2017dzd}, is that one must know the all the spinning conformal blocks of the 
conformal group ${\rm SO}(5,1)$ and all tensor structures to determine the coefficients $\T$ in the dynamical system, or 
one must have another way to determine them -- perhaps directly from the representation theory of the conformal group. 

\medskip

{\bf Acknowledgements:} This paper is based on talks given at the conferences ``In Memoriam Rudolph Haag'' (Hamburg, September 2016)
and the ``Memorial Symphosium in Honor of W. Zimmermann'' (Munich, May 2016) where these results were announced and outlined.
I am grateful for support and hospitality at those meetings. I am also grateful to S. Rychkov 
for several discussions on the topic of this paper. Parts of this work were completed during the program ``Conformal Field Theories and Renormalisation Group Flows in dimensions $d>2$'' at the Galileo-Galilei-Institute in 2016. Hospitality is gratefully acknowledged. 
I would also like to thank M. Fr\" ob, J. Holland and Y. Honma for valuable discussions at an early state of this work. 

The research leading to these results has received funding from the European Research Council under the European Union's Seventh Framework Programme (FP7/2007-2013) / ERC grant agreement no QC \& C 259562.


\end{document}